\documentclass[a4 paper, 11pt]{article}
\usepackage{amsmath}
\usepackage{latexsym}
\usepackage{amsfonts}
\usepackage{amssymb}
\usepackage{amsthm}
\usepackage{amsbsy}
\usepackage{graphicx}
\usepackage{color}
\usepackage{epstopdf}
\usepackage{todonotes}
\usepackage{fullpage}
\usepackage{hyperref}
\usepackage{mathrsfs}
\numberwithin{equation}{section}
\usepackage{float}
\usepackage[left=2cm, right=4cm, top=2cm, bottom=2cm]{geometry}
\newcommand{\bbb}{\color{black}}

\newcommand{\R}{\mathbb{R}}

\newcommand{\cE}{\mathcal{E}}
\newcommand{\cF}{\mathcal{F}}
\newcommand{\cG}{\mathcal{G}}
\newcommand{\cP}{\mathcal{P}}

\newcommand{\eps}{\varepsilon}
\renewcommand{\div}{\operatorname{div}}
\newcommand{\curl}{\operatorname{curl}}

\newcommand{\id}{\mathbf{id}}

\theoremstyle{definition}
\newtheorem{theorem}{Theorem}[section]
\newtheorem{lemma}[theorem]{Lemma}
\newtheorem{proposition}[theorem]{Proposition}
\newtheorem{definition}[theorem]{Definition}
\newtheorem{corollary}[theorem]{Corollary}
\newtheorem{remark}[theorem]{Remark}
\newtheorem{example}[theorem]{Example}
\theoremstyle{plain}

\begin{document}
	\title{\Large \bf On the single field formulation in magnetostatics
	}

	\author{
		Stefan Kr\"{o}mer
		\thanks{The Czech Academy of Sciences, Institute of Information Theory and Automation, Prague, Czech Republic} \and
		Giuseppe Tomassetti
		\thanks{Department of Industrial, Electronic, and Mechanical Engineering - Roma Tre University, Roma, Italy}
		\thanks{Institute for Applied Mathematics ``Mauro Picone'' (IAC) - National Research Council (CNR), Rome, Italy}
	}

	\maketitle
	\begin{abstract}
		We systematically discuss the equivalence of
		two variational formulations of magnetostatics, in terms of magnetization and magnetic field on the one hand and the single field formulation
		using only magnetic induction.	
		To demonstrate that this link is stable also when the magnetic laws are coupled with other variational static models, elasticity is included in the models as well.
		Interestingly, despite the fact that 
		the corresponding magnetoelastic energy densities in the material can be computed via Legendre-Fenchel transform in the magnetic state variables, 
		the two formulations are not linked by standard convex duality on the level of the functionals. 
		In addition, convexity and coercivity of the given functional are neither required for the transformation nor always preserved by it.
	\end{abstract}


	\section{Introduction}
	Within the scope of Brown's theory \cite{Brown1966MagnetoelasticInteractions}, the configuration of a magnetoelastic body $\Omega$ is described by a \emph{deformation} $\mathbf y:\Omega\to\mathbb R^3$, a \emph{magnetization density} $\mathbf{m}:\mathbf{y}(\Omega)\to\mathbb R^3$, and a \emph{self magnetic field} $\mathbf{h}_{\rm s}:\mathbb R^3\to\mathbb R^3$. The self field $\mathbf{h}_{\rm s}$ is determined by the magnetization density $\mathbf{m}$ as the distributional solution of the \emph{Maxwell constraint}:
	\begin{equation}\label{divcon}
		\begin{aligned}
			&\operatorname{curl}\mathbf{h}_{\rm s}=0,\\
			&\operatorname{div}(\mathbf{h}_{\rm s}+\chi_{\mathbf{y}(\Omega)}\mathbf{m})=0,
		\end{aligned}\quad\text{ on }\mathbb R^3,
	\end{equation}
	where $\chi_{\mathbf{y}(\Omega)}\mathbf{m}$ is the trivial extension of $\mathbf{m}$ to $\mathbb R^3$.  
When a current density $\mathbf{j}:\mathbb R^3\to\mathbb R^3$
is generated by an external source outside $\Omega$,
the body is acted upon by an \emph{applied magnetic field}
$\mathbf{h}_{\rm a}:\mathbb R^3\to\mathbb R^3$, given by
\begin{equation}\label{eq:ha}
    \begin{aligned}
        \operatorname{curl}\mathbf{h}_{\rm a}=\mathbf{j},\\
        \operatorname{div}\mathbf{h}_{\rm a}=0,
    \end{aligned}\quad\text{ in }\mathbb R^3,
\end{equation}
and its stable equilibria are local minimizers, under the constraint \eqref{divcon}, of the
\emph{magneto-elastic energy}:
	\begin{equation}\label{eq:JK}
		\widehat{\mathcal E}(\mathbf{y}, \mathbf{m},\mathbf{h}_{\rm s})
		=\int_{\Omega}\widehat\Phi(\mathbf{x},\nabla \mathbf{y}, \mathbf{m}\circ \mathbf{y}) {\rm d}\mathbf{x}
		+\frac{\mu_0}2\int_{\mathbb R^3}|\mathbf{h}_{\rm s}|^2{\rm d}\boldsymbol{\xi}
		-\mu_0\int_{\mathbf{y}(\Omega)}\mathbf{h}_{\rm a}\cdot\mathbf{m}\,{\rm d}\boldsymbol{\xi},
	\end{equation}
	which we write by adapting the format used by James and Kinderlehrer in their investigation of magnetic microstructures in magnetoelastic crystalline bodies undergoing large strains \cite{JamesKinderlehrer1993Magnetostriction}. 
	
Equation \eqref{eq:ha} is set in the framework of magnetostatics, 
which requires two conditions on the current density $\mathbf{j}$: 
that induces no charge accumulation, $\operatorname{div}\mathbf{j}=0$ 
, and that its time variations are slow 
compared to the time light takes to traverse the body, the so-called quasi-static regimes. The first condition is necessary for 
\eqref{eq:ha} to be consistent, since 
$\operatorname{div}(\operatorname{curl}\mathbf{h}_{\rm a})=0$ 
identically. The second justifies neglecting the displacement 
current in Amp\`ere's law \cite{jackson}.

	In writing \eqref{eq:JK} we use the notation
		\begin{equation}
		(\mathbf{m} \circ \mathbf{y})(\mathbf{x})=\mathbf{m}[\mathbf{y}(\mathbf{x})],\qquad \mathbf{x}\in\Omega,
	\end{equation}
	which emphasizes that, while the domain of definition of the deformation $\mathbf{y}$ is the \emph{reference configuration} $\Omega$, the magnetization density is defined in the \emph{deformed configuration} $\mathbf{y}(\Omega)$. Thus, when integrating the energy density $\widehat\Phi$ over the reference configuration, taking the referential point $\mathbf{x}$ as independent variable, the magnetization density must be evaluated at $\boldsymbol{\xi}=\mathbf{y}(\mathbf{x})$.	
	
	Within this framework, the (total) \emph{magnetic field}, 
	\begin{equation}
		\mathbf{h}= \mathbf{h}_{\rm a} + \mathbf{h}_{\rm s},
	\end{equation}
	is recovered as the sum of the applied- and self-field. The magnetic induction, which is the quantity that can be actually measured, is given by
	\begin{equation}\label{eq:induction}
		\mathbf{b}=\mu_0(\mathbf{m}+\mathbf{h}).
	\end{equation}

	Brown's original formulation was developed to model ferromagnetic materials at the microscopic level where domain structures can be observed \cite{Brown1940TheoryApproachSaturation,Brown1941EffectDislocations,Brown1945VirtuesDomain,Brown1957CriterionUniform}. At this level of spatial resolution, the \emph{specific magnetization density} $\mathbf{m}_s=(\operatorname{det}\nabla\mathbf{y})\mathbf{m}$ in a ferromagnetic body obeys the saturation constraint  $|\mathbf{m}|=m_s$.  The non-convexity arising from such constraint leads to the possibility of multiple local minima and formation of microstructures, whose analysis can be dealt rigorously with relaxation and Young's measures \cite{DeSimone1993EnergyminimisersLargeFerromagnetic,KruzikProhl2001YoungMeasureApproximation} or by introducing exchange energy, which penalizes spatial variations of $\mathbf{m}$, as discussed in the review \cite{DeSimoneKohnMullerOtto2006_AnalyticalMicromagnetics}. 
	
	Kankanala and Triantafyllidis \cite{KankanalaTriantafyllidis2004MRE} leveraged Brown's approach in situations beyond classical ferromagnets, such as soft magnetoelastic bodies --- specifically, magnetorheological elastomers. Although they used a similar mathematical formulation, their physical perspective is different from Brown's: in their work, the field $\mathbf{m}$ represents  the magnetization density at an observation scale much larger than that of magnetic domains. In particular, the saturation constraint is no longer enforced, and, more importantly, the detailed spatial profile of the magnetization is not the quantity of main interest. This second aspect perhaps explains why, in the context of soft magnetoelastic bodies, alternative formulations have emerged, which replace magnetization with other magnetic variables. The first of these alternative formulation was proposed by Dorfmann and Ogden \cite{DorfmannOgden2004Magnetoelastic}. In their formulation of magneto-elasticity the magnetic state of the body is described by the \emph{Lagrangian magnetic induction} $\mathbf{B}$, which is related to the (Eulerian) magnetic induction \eqref{eq:induction} by the following transformation formula:
	\begin{equation}\label{langrangeB}
		\mathbf{B}={J}{\mathbf{F}}^{-1}(\mathbf{b}\circ{\mathbf{y}}),
	\end{equation}
	where
	\begin{equation}\label{eq:FJ}
		\mathbf{F}=\nabla\mathbf{y},\qquad J=\det\mathbf{F}.
	\end{equation}
	The Dorfmann \& Ogden theory admits a variational formulation, which has been discussed by Bustamante et al. in \cite{h}. This formulation has been considered, with minor changes, by \v{S}ilhav\'y \cite{b} in his discussion of existence theorem in non-linear electro-magneto-elasticity. With suitable modifications to align with the form of \eqref{eq:JK}, the functional form of the energy considered by \v{S}ilhav\'y can be written (up to a constant) as
	\begin{equation} \label{E-Silhavy}
		\overline{\mathcal E}(\mathbf{y},\mathbf{B})=\int_{\Omega}\overline\Phi(\mathbf{x},\mathbf{F},\mathbf{B})\,{\rm d}\mathbf{x}+\frac 1 {2\mu_0}\int_{\mathbb R^3}\left\{\frac 1 J |\mathbf{F}(\mathbf{B}-J(\mathbf{b}_{\rm a}\circ\mathbf{y}))|^2 \right\}{\rm d}\mathbf{x},
	\end{equation}
	where $\mathbf{y}$ has been extended to $\mathbb R^3$ in a suitable way.
	
	Equilibrium states are identified with the critical points of $\overline{\mathcal E}$ under the constraint
	\begin{equation}
		\operatorname{Div}\mathbf{B}=0,
	\end{equation}
	which thanks to the definition \eqref{langrangeB} is equivalent to the \emph{solenoidal constraint}:
	\begin{equation}\label{eq:divb0}
		\operatorname{div}\mathbf{b}=0.
	\end{equation}
	Working with the Lagrangian induction field $\mathbf{B}$ bypasses several difficulties inherent with the fact that the magnetic state variables are naturally defined in the physical space, here represented by $\mathbb R^3$, whereas the domain of definition of the deformation is a subdomain $\Omega$ of a fictitious space still modeled with $\mathbb R^3$, thus simplifying actual calculations, be they analytical or numerical. However, it introduces the conceptual difficulty of extending the deformation gradient outside $\Omega$, which is not always possible, since a sufficiently regular extension of the deformation may not exist. Since in this paper we are concerned with theoretical aspects, we shall consider critical points, subjected to the constraint \eqref{eq:divb0}, of the following functional:
	\begin{equation}\label{eq:bb}
		\mathcal E(\mathbf{y},\mathbf{b})=\int_{\Omega}\Phi(\mathbf{x},\nabla\mathbf{y},\mathbf{b}\circ\mathbf{y}){\rm d}\mathbf{x}+\frac{1}{2\mu_0}\int_{\mathbb R^3}|\mathbf{b}-\mathbf{b}_{\rm a}|^2{\rm d}\boldsymbol{\xi}.
	\end{equation}
If the function $\Phi$ is related to $\overline\Phi$ by the transformation formula
\begin{equation}
\Phi(\mathbf{x},\mathbf{F},\mathbf{b})=\overline\Phi(\mathbf{x},\mathbf{F},J\mathbf{F}^{-1}\mathbf{b}),
\end{equation}
then by \eqref{langrangeB}, $\mathcal E$ and $\overline{\mathcal E}$ coincide.

For many cases of engineering interest, at least one of the two functionals $\mathbf{b}\mapsto\mathcal E(\mathbf{y},\mathbf{b})$ and $(\mathbf{m},\mathbf{h}_{\rm s})\mapsto\widehat{\mathcal E}(\mathbf{y},\mathbf{m},\mathbf{h}_{\rm s})$ is convex and coercive in their respective magnetic variables. In particular, for any fixed $\mathbf{y}$ that functionals then admits a unique minimizer. However, while the formulation based on \eqref{eq:bb} is local, the one based on $\mathbf{m}$ and $\mathbf{h}_{\rm s}$ is not. In fact, by writing
\begin{equation}
\mathbf{b}=\operatorname{curl} \mathbf{a} \quad \text { in } \mathbb{R}^3,
\end{equation}
and adopting the vector potential $\mathbf{a}$ as the unknown, the solenoidal constraint \eqref{eq:divb0} is enforced identically, up to the usual gauge invariance $\mathbf{a} \mapsto \mathbf{a}+\nabla u$, and the resulting Euler-Lagrange equations are pointwise. By contrast, even if the irrotationality constraint $\operatorname{curl}\mathbf{h}_{\rm s} = 0$ in the $(\mathbf{m},\mathbf{h})$-based formulation is enforced by introducing a scalar potential $\varphi$ and by setting $\mathbf{h}_{\rm s} = -\nabla\varphi$, the dependence of $\varphi$ from $\mathbf{m}$ is non-local. 

The equivalence between the two formulations has been discussed in several papers in the engineering literature \cite{h,Danas2017EffectiveResponse,ShaSa21a}. Although all of these papers, either implicitly or explicitly acknowledge the role convex duality, the only rigorous discussion has been given, to our knowledge, by Pedregal \& Yan \cite{PedregalYan2010DualityMicromagnetics}. 
More precisely, \cite{PedregalYan2010DualityMicromagnetics} only considers the special case of a coercive micromagnetics energy density with a hard saturation constraint, a nonconvex but still coercive \bbb case where full model equivalence can only be expected after a suitable relaxation. Interestingly, if the hard saturation constraint  or coercivity \bbb are dropped, the link to the single-field model can change significantly, see Example~\ref{ex:saturation} 
(close to \cite{PedregalYan2010DualityMicromagnetics}, but when starting from the $\mathbf{b}$-model, the hard saturation constraint is never fully recovered when passing to the $(\mathbf{m},\mathbf{h})$-model, and softer versions of saturation are not covered in \cite{PedregalYan2010DualityMicromagnetics} at all),
and Example~\ref{ex:diapara} (demonstrating that in the $(\mathbf{m},\mathbf{h})$-formulation, diamagnetic materials naturally lead to concave functionals that should be maximized, not minimized).
Both are natural cases that we cover here as well. 

The goal of the present paper is to provide a precise and rigorous discussion of the equivalence of the two formulations, which goes beyond the related results of \cite{PedregalYan2010DualityMicromagnetics}. With this, we hope to fill what we believe to be a gap in the existing literature, where the precise meaning of such equivalence is not stated clearly. In particular, we wish to make the point that a more careful treatment is needed, compared to the existing literature cited above. In doing so, we still want to keep the level of mathematical technicality  accessible to the engineering community. 
In order to outline our line of arguments, we present our equivalence statement in a special case, and then we comment on the result. For simplicity, we restrict ourselves to the smooth case, i.e., where all densities are smooth, and we assume that the material energy density $\widehat\Psi(\mathbf x,\mathbf F,\mathbf m)$ is convex with respect to $\mathbf m$.

\paragraph{Equivalence statement for the smooth convex case.}
In discussing the equivalence between the two formulations, the deformation $\mathbf{y}:\Omega\to\mathbb{R}^3$ does not interfere at all with the equivalence of the magnetic models, and hence can be treated as a fixed parameter. We write $\mathcal{E}_{\mathbf{y}}(\mathbf{b}):=\mathcal{E}(\mathbf{y},\mathbf{b})$
and $\widehat{\mathcal{E}}_{\mathbf{y}}(\mathbf{m},\mathbf{h}_{\rm s})
:=\widehat{\mathcal{E}}(\mathbf{y},\mathbf{m},\mathbf{h}_{\rm s})$
for the two energy functionals viewed as functionals of the magnetic variables only.

The key hypothesis for the equivalence of the two formulations is that the energy densities $\Phi$ and $\widehat\Phi$ are linked by the following transformation formula:
\begin{equation}\label{relations1}
		\Phi(\mathbf{x},\mathbf{F},\mathbf{b})=-\widehat\Psi^*(\mathbf{x},\mathbf{F},\mathbf{b}).
	\end{equation}
	Here, $\widehat\Psi$ is the \emph{augmented material energy density}, given by
	\begin{equation}\label{def:hatPsi}
		\widehat\Psi(\mathbf{x},\mathbf{F},\mathbf{m}):=\widehat\Phi(\mathbf{x},\mathbf{F},\mathbf{m})+\frac {\mu_0}2|\mathbf{m}|^2,
	\end{equation}
	where $\widehat\Psi^*$ is the Legendre-Fenchel transform of $\widehat\Psi$ in the variable $\mathbf{m}$, given by \eqref{def:LeFe-convex}
	(still in the smooth cases, for energies with invertible gradient,  $\widehat\Psi^*$ will be replaced by the generalized Legendre-Fenchel transform $\widehat\Psi^\diamond$ defined in \eqref{genFc-smooth}). Accordingly, $\Phi(\mathbf x,\mathbf F,\mathbf b)$ is concave in $\mathbf{b}$. Moreover, the applied fields are linked by $\mathbf{h}_{\rm a}=\mu_0 \mathbf{b}_{\rm a}$. \bbb

The reverse transformation corresponding to \eqref{relations1} is
	\begin{equation}\label{relations1b}
		\widehat\Phi(\mathbf{x},\mathbf{F},\mathbf{m})= (-\Phi)^*(\mathbf{x},\mathbf{F},\mathbf{m})-\frac {\mu_0}2|\mathbf{m}|^2,
	\end{equation}
with the Legendre-Fenchel transform now acting in the variable $\mathbf{b}$ of $\Phi$. 
Here we omit some technical assumptions on $\Omega$, $\mathbf{y}$, $\widehat\Phi$ and $\Phi$ that we shall specify in the body of the paper. With these assumption, the equivalence statement we prove can be stated as follows in the simplest case:

\begin{itemize}
\item[(i)] If $(\overline{\mathbf{m}},\overline{\mathbf{h}}_{\rm s})$ is a critical point of
$\widehat{\mathcal{E}}_{\mathbf{y}}$ under the Maxwell constraint \eqref{divcon},
then the field $\overline{\mathbf{b}}$ defined by
\begin{align}\label{eq:induction2}
	\overline{\mathbf{b}}=\mathbf{b}_{\rm a}
	+\mu_0(\chi_{\mathbf{y}(\Omega)}\overline{\mathbf{m}}+\overline{\mathbf{h}}_{\rm s})
	\qquad\text{in }\mathbb R^3.
\end{align}
is a critical point of $\mathcal{E}_{\mathbf{y}}$ under the solenoidal constraint
\eqref{eq:divb0}, with $\Phi$ given by \eqref{relations1}. Moreover,
\begin{equation}\label{eq:eneq}
\mathcal{E}_{\mathbf{y}}(\overline{\mathbf{b}})
=\widehat{\mathcal{E}}_{\mathbf{y}}(\overline{\mathbf{m}},\overline{\mathbf{h}}_{\rm s}).
\end{equation}

\item[(ii)] Conversely, if $\overline{\mathbf{b}}$ is a critical point of $\mathcal{E}_{\mathbf{y}}$
under the solenoidal constraint \eqref{eq:divb0}, then for
$\overline{\mathbf{m}}$ defined via the constitutive relation 
\begin{equation}\label{m-from-b-smooth-intro}
\begin{aligned}
	\overline{\mathbf{m}}\circ\mathbf y &=-\nabla_{\mathbf{b}}\Phi(\cdot,\mathbf F,\overline{\mathbf{b}}\circ\mathbf y) &&\text{in } \Omega,\\
\end{aligned}
\end{equation}
and $\overline{\mathbf{h}}_{\rm s}$ is the unique solution of \eqref{divcon}, equivalently, \bbb  the curl-free part of $-\chi_{\mathbf{y}(\Omega)}\mathbf{m}$ 
from the Helmholtz decomposition on $\R^3$ (cf.~\eqref{mh-from-b-smooth}), 
the pair $(\overline{\mathbf{m}},\overline{\mathbf{h}}_{\rm s})$ is a critical point of
$\widehat{\mathcal{E}}_{\mathbf{y}}$ under the Maxwell constraint \eqref{divcon},
with $\widehat{\Phi}_{\mathbf{y}}$ given by \eqref{relations1b}. Moreover, \eqref{eq:eneq} and \eqref{eq:induction2} are again satisfied.
\end{itemize}
From the above equivalence statement, it emerges that it is crucially important to distinguish equilibrium points from general admissible states. Precisely, given a generic state $(\mathbf{m},\mathbf{h}_{\rm s})$ that does not satisfy the Euler-Lagrange equations, the corresponding $\mathbf{b}$ defined by \eqref{eq:induction2} does not satisfy 
$\mathcal{E}_{\mathbf{y}}(\mathbf{b})=\widehat{\mathcal{E}}_{\mathbf{y}}(\mathbf{m},\mathbf{h}_{\rm s})$ as can be easily checked for concrete examples, for instance the diamagnetic case of Example 4.1 where
$\mathcal{E}_{\mathbf{y}}$ is convex and coercive whereas $\widehat{\mathcal{E}}_{\mathbf{y}}$ is not bounded from below. Conversely, if $\mathbf{b}$ is a generic admissible state that does not satisfy the Euler-Lagrange equations, the corresponding $\mathbf{m}$ defined by \eqref{m-from-b-smooth-intro} and the corresponding $\mathbf{h}_{\rm s}$ obtained by solving \eqref{divcon} do not satisfy $\mathcal{E}_{\mathbf{y}}(\mathbf{b})=\widehat{\mathcal{E}}_{\mathbf{y}}(\mathbf{m},\mathbf{h}_{\rm s})$ either. In other words, the transformation linking the two models is such that the energies coincide only at equilibrium. In our opinion, such distinction has not been sufficiently emphasized in the engineering literature.

Fully rigorous versions of the above equivalence statement including precise assumptions are presented in
Theorem~\ref{thm:eq-smooth} and Corollary~\ref{cor:eneq-smooth} (smooth case,
Section~\ref{sec:smooth}), and in Theorems~\ref{thm:mh-to-b} and~\ref{thm:b-to-mh} and Corollary~\ref{cor:eneq-cc} (convex/concave case,
Section~\ref{sec:cc}). These results also discuss cases beyond the convex one. Examples are given in Section~\ref{sec:ex}.

\section{The smooth case}\label{sec:smooth}
Throughout this section and the rest of the paper, we regard the deformation $\mathbf{y}:\Omega\to\mathbb{R}^3$ as fixed and we make the assumption that the deformation is regular enough so that the image $\mathbf{y}(\Omega)$ is a measurable subset of $\R^3$. As to the applied field, 
we assume that 
$\mathbf{b}_a\in L^2(\R^3;\R^3)$ with $\div \mathbf{b}_a=0$ on $\R^3$.
Since these assumption are common to all theorem statements, we shall leave them tacit.
The two energy functionals that we compare are \eqref{eq:bb} and \eqref{eq:JK}, but now the domain of integration is the deformed configuration 
after the 
change 
of variables $\boldsymbol{\xi}= \mathbf{y}(\mathbf{x})$:
\begin{align}
  \cE_{\mathbf{y}}(\mathbf{b})
    &:= \int_{\mathbf{y}(\Omega)} \Phi_{\mathbf{y}}(\boldsymbol{\xi},\mathbf{b})\,{\rm d}\boldsymbol{\xi}
      + \frac{1}{2\mu_0}\int_{\mathbb R^3}|\mathbf{b}-\mathbf{b}_{\rm a}|^2
        \,{\rm d}\boldsymbol{\xi},
    \label{E(b)}\\
  \widehat\cE_{\mathbf{y}}(\mathbf{m},\mathbf{h}_{\rm s})
    &:= \int_{\mathbf{y}(\Omega)}\widehat\Phi_{\mathbf{y}}(\boldsymbol{\xi},\mathbf{m})
        \,{\rm d}\boldsymbol{\xi}
      + \frac{\mu_0}{2}\int_{\mathbb R^3}|\mathbf{h}_{\rm s}|^2\,{\rm d}\boldsymbol{\xi}
      - \int_{\mathbf{y}(\Omega)}\mathbf{m}\cdot\mathbf{b}_{\rm a}\,{\rm d}\boldsymbol{\xi},
    \label{E(mh)}
\end{align}
Here the transformed material energy densities are
\begin{align}
  \Phi_{\mathbf{y}}(\boldsymbol{\xi},\mathbf{b})
    &:= \frac{1}{J(\mathbf{y}^{-1}(\boldsymbol{\xi}))}
        \Phi\!\left(\mathbf{y}^{-1}(\boldsymbol{\xi}),\,F(\mathbf{y}^{-1}(\boldsymbol{\xi})),\,\mathbf{b}\right),
    \label{eq:Phiy-def}\\
  \widehat\Phi_{\mathbf{y}}(\boldsymbol{\xi},\mathbf{m})
    &:= \frac{1}{J(\mathbf{y}^{-1}(\boldsymbol{\xi}))}
        \widehat\Phi\!\left(\mathbf{y}^{-1}(\boldsymbol{\xi}),\,F(\mathbf{y}^{-1}(\boldsymbol{\xi})),\,\mathbf{m}\right).
    \label{eq:hatPhiy-def}
\end{align}

In this section we will show that the two models, the one based on the energy $\widehat{\mathcal E}_{\mathbf{y}}(\mathbf m,\mathbf h_{\rm s})$ and the one based on the energy $\cE_{\mathbf{y}}(\mathbf b)$, are indeed equivalent provided that the energy densities $\widehat{\Phi}_{\mathbf{y}}(\boldsymbol{\xi},\mathbf m)$ and $\Phi_{\mathbf{y}}(\boldsymbol{\xi},\mathbf b)$ are related by a suitable duality transformation (see \eqref{H0}, \eqref{H0p} and \eqref{eduality2} below). Such equivalence holds both at the level of stationary points and at the level of energies. 

For clarity of exposition, we divide this section into three distinct subsections. In the first one, we give the notion of constrained critical point. In the second one, we discuss the correspondence between constrained critical points of the two models at the level of the Euler-Lagrange equations. In the third one, we show the energies at the two models coincide at critical points.

\subsection{Constrained critical points}
To begin, we recapitulate the Euler-Lagrange equations of the two models. These equations are well known, and their derivation is given in Proposition~\ref{prop:EL-smooth} in the Appendix for the reader's convenience. If the energy functionals are differentiable, states that satisfy the Euler-Lagrange equations are natural candidates for equilibrium states, and we call them \emph{smooth constrained critical points}.

\begin{definition}[Smooth constrained critical points]\label{def:crit-smooth}~
\begin{enumerate} 
    \item[(i)] We say that $(\overline{\mathbf{m}},\overline{\mathbf{h}}_{\rm s})\in L^2(\Omega;\R^3)\times L^2(\R^3;\R^3)$ is a smooth constrained critical point of the functional $\widehat\cE_{\mathbf{y}}$ defined in \eqref{E(mh)} if the Maxwell equations \eqref{divcon} hold:
    \begin{equation}
    \operatorname{div}(\chi_{\mathbf{y}(\Omega)}\overline{\mathbf{m}}+\overline{\mathbf{h}}_{\rm s})=0, \quad \operatorname{curl}\overline{\mathbf{h}}_{\rm s}=0 \quad \text{in $\mathbb R^3$},
    \end{equation}
    and the following Euler-Lagrange equation holds:
    \begin{equation}\label{equil-mh-smooth} 
			\nabla_{\mathbf{m}} \widehat{\Phi}_{\mathbf{y}}\left(\cdot,\overline{\mathbf{m}}\right)=\mu_0\overline{\mathbf{h}}_{\rm s}+ \mathbf{b}_{\rm a} 
			\quad \text { a.e. in } \mathbf{y}(\Omega).
    \end{equation}
    \item[(ii)] A field $\overline{\mathbf b}\in L^2(\R^3;\R^3)$ is a smooth constrained critical point
    of the functional $\cE_{\mathbf y}$ defined in \eqref{E(b)} if $\operatorname{div}\overline{\mathbf{b}}=0$, and there exists a scalar potential $\overline\varphi\in \mathring W^{1,2}(\R^3)$
		(the space defined in Subsection~\ref{secA:Helmholtz}) such that the following Euler-Lagrange equation holds:
    \begin{equation}\label{equil-b-smooth}
    	-\chi_{\mathbf{y}(\Omega)}\nabla_{\mathbf{b}}\Phi_{\mathbf{y}}(\cdot,\overline{\mathbf{b}})=\frac{\overline{\mathbf b}-\mathbf b_a}{\mu_0}+\nabla\overline\varphi\qquad\text{a.e. in }\mathbb R^3\,.
    \end{equation}
\end{enumerate}
\end{definition}
For a derivation of these Euler-Lagrange equations see Appendix~\ref{secC:crit}. 

\begin{remark}
An equivalent way of writing the Euler-Lagrange equation \eqref{equil-mh-smooth} is
\begin{equation}\label{equil-mh-smooth-2}
\nabla_{\mathbf{m}} \widehat{\Psi}_{\mathbf{y}}\left(\cdot,\overline{\mathbf{m}}\right)=\mu_0(\overline{\mathbf m}+\overline{\mathbf{h}}_{\rm s})+ \mathbf{b}_{\rm a} 
\quad \text { a.e. in } \mathbf{y}(\Omega),
\end{equation}
where $\widehat\Psi_{\mathbf y}(\boldsymbol{\xi},\mathbf m)=\widehat\Phi_{\mathbf y}(\boldsymbol{\xi},\mathbf m)+\frac{\mu_0}{2}|\mathbf m|^2$ is the augmented material energy density defined in \eqref{def:hatPsi}. The two forms are equivalent. However, the form \eqref{equil-mh-smooth} is more convenient for the comparison with the $\mathbf{b}$-based model, since the right-hand side is the magnetic induction $\mathbf{b}$, this equation can be interpreted as a constitutive relation linking $\mathbf{b}$ to $\mathbf{m}$. 	
\end{remark}

\begin{remark}
Equation \eqref{equil-b-smooth} is often written in the form
\begin{equation}\label{equil-b-smooth-curl}
\curl\left(\chi_{\mathbf{y}(\Omega)}\nabla_{\mathbf{b}}\Phi_{\mathbf{y}}(\cdot,\overline{\mathbf{b}})+\frac{\overline{\mathbf b}-\mathbf b_a}{\mu_0}\right)=0.
\end{equation}
The two forms are equivalent. However, while \eqref{equil-b-smooth} holds pointwise (almost everywhere), \eqref{equil-b-smooth-curl} holds only in the sense of distributions. \end{remark}

\begin{remark}
In \eqref{equil-b-smooth}, the field $\nabla\overline\varphi$ is a Lagrange multiplier associated with the divergence-free constraint on $\mathbf{b}$. Such multiplier is not present in the Euler-Lagrange equation \eqref{equil-mh-smooth} of the $\mathbf{m}$-$\mathbf{h}$-based model, since the constraints \eqref{divcon} are automatically satisfied through the projection of $\chi_{\mathbf y(\Omega)}$ on the space of curl-free vector fields, as explained in Remark~\ref{rem:h-from-m}.
\end{remark}

\subsection{Correspondence between constrained critical points}
In this subsection we show that the two models are indeed equivalent at the level of constrained critical points, provided that the energy densities $\widehat{\Phi}_{\mathbf{y}}$ and $\Phi_{\mathbf{y}}$ are related by the duality relations \eqref{H0}. In particular, if $(\overline{\mathbf{m}},\overline{\mathbf{h}}_{\rm s})$ is a constrained critical point of $\widehat\cE_{\mathbf{y}}$, then the field $\overline{\mathbf{b}}$ corresponding to the physical definition
\begin{align}\label{b-from-mh-smooth}
   \overline{\mathbf{b}}:=\mathbf{b}_{\rm a}+\mu_0(\chi_{\mathbf{y}(\Omega)}\overline{\mathbf{m}}+\overline{\mathbf{h}}_{\rm s})
\end{align}
is a constrained critical point of $\cE$. Going from the $\mathbf b$-model to the $\mathbf{m}$-$\mathbf{h}$-model, is slightly more involved and requires instead to define $(\overline{\mathbf{m}},\overline{\mathbf{h}}_{\rm s})$ in terms of $\overline{\mathbf{b}}$ by:
\begin{equation}\label{mh-from-b-smooth}
\begin{aligned}
	\overline{\mathbf{m}} &:=-\nabla_{\mathbf{b}}\Phi_{\mathbf{y}}(\cdot,\overline{\mathbf{b}}) &&\text{in } \mathbf{y}(\Omega),\\
	\overline{\mathbf{h}}_{\rm s} &:=-\cP[\chi_{\mathbf{y}(\Omega)}\overline{\mathbf{m}}] &&\text{in } \mathbb{R}^3.
\end{aligned}
\end{equation}
The first of \eqref{mh-from-b-smooth} depends on the material, and can be interpreted as a constitutive relation between magnetic induction and magnetization. In the second equation $\mathcal P[\cdot]$ is the orthogonal projection of $L^2(\mathbb R^3;\mathbb R^3)$ on the space of curl-free vector fields. 

\begin{remark}
The Helmholtz decomposition of $L^2(\mathbb R^3;\mathbb R^3)$ based on $\mathcal P$ (see Appendix \ref{secA:Helmholtz}) reveals a partial link between \eqref{b-from-mh-smooth} and \eqref{mh-from-b-smooth}. It allows us to decompose any $L^2(\mathbb R^3;\mathbb R^3)$ vector field into the sum of a curl-free and a divergence-free vector field, in the present case reads
\begin{equation}
\chi_{\mathbf{y}(\Omega)}\overline{\mathbf{m}}=\cP[\chi_{\mathbf{y}(\Omega)}\overline{\mathbf{m}}]+(I-\cP)[\chi_{\mathbf{y}(\Omega)}\overline{\mathbf{m}}].
\end{equation}
Assuming that \eqref{b-from-mh-smooth} with $\overline{\mathbf{b}}$ divergence-free and $\overline{\mathbf{h}}_{\rm s}$ curl-free, we 
necessarily have (since $\mathbf b_{\rm a}$ is divergence-free by assumption)
\begin{equation}\label{eq:mh-from-b-smooth-2}
\cP[\chi_{\mathbf{y}(\Omega)}\overline{\mathbf{m}}]=-\overline{\mathbf{h}}_{\rm s}
\quad\text{and}\quad (I-\cP)[\chi_{\mathbf{y}(\Omega)}\overline{\mathbf{m}}]=\frac{1}{\mu_0}(\overline{\mathbf{b}}-\mathbf{b}_{\rm a}).
\end{equation}
\end{remark}

The precise statement of the correspondence between critical points is in Theorem \ref{thm:eq-smooth} below, which we prove under the following assumptions. In particular,
they link the material energy densities $\widehat{\Phi}_{\mathbf y}$ and $\Phi_{\mathbf y}$ by suitable duality relations, where the 
symbol $\diamond$ denotes the generalized Legendre-Fenchel transform defined in \eqref{genFc-smooth}. 
If we are coming from the $(\mathbf{m},\mathbf{h}_s)$-based formulation with given $\widehat\Phi_{\mathbf{y}}$, 
define the corresponding density function as 
\begin{align*}
	&\begin{aligned}[c]
			&\Phi_{\mathbf{y}}(\boldsymbol{\xi},\mathbf{b})=-\widehat\Psi_{\mathbf{y}}^\diamond(\boldsymbol{\xi},\mathbf{b}),\quad
			\text{where $\widehat\Psi_{\mathbf y}(\boldsymbol{\xi},\mathbf m)=\widehat\Phi_{\mathbf y}(\boldsymbol{\xi},\mathbf m)+\frac{\mu_0}{2}|\mathbf m|^2$ (cf.~\eqref{def:hatPsi})},
		\end{aligned}\label{H0}\tag{H0}\\
\intertext{and require that}
	&\begin{aligned}[c]
		&\parbox{.84\textwidth}{
			$(\boldsymbol{\xi},\mathbf{m})\mapsto \widehat\Phi_{\mathbf{y}}(\boldsymbol{\xi},\mathbf{m})$, 
			$\widehat\Phi_{\mathbf{y}}:\mathbf{y}(\Omega)\times \R^3\to \R\cup\{+\infty\}$, \\
			is measurable in $\boldsymbol{\xi}$ for each $\mathbf{m}$
			and lower semi-continuous in $\mathbf{m}$ for a.e.~$\boldsymbol{\xi}$;
		}
	\end{aligned}	\tag{H1}\label{H1} \\
  &\begin{aligned}[c]
		&\text{$-C|\mathbf{m}|^2-C \leq \widehat\Phi_{\mathbf{y}}(\boldsymbol{\xi},\mathbf{m})$ for all $\mathbf{m}$ and a.e.\ $\boldsymbol{\xi}$, with a constant $C>0$;}
	\end{aligned}\tag{H2}\label{H2} \\
	&\begin{aligned}[c] 
		\parbox{.84\textwidth}{
			for a.e. $\boldsymbol{\xi}\in\mathbf y(\Omega)$, 
			$\widehat\Phi_{\mathbf{y}}(\boldsymbol{\xi},\cdot)$ is differentiable 
			and its gradient
			$\mathbf{m}\mapsto \nabla_{\mathbf{m}}\widehat\Phi_{\mathbf{y}}(\boldsymbol{\xi},\mathbf{m})$, $\R^3\to \R^3$,
			is a continuous and continuously invertible function.
		}	
	\end{aligned}\label{H3}\tag{H3}
\end{align*} 
If we are the given $\mathbf{b}$-based formulation featuring $\Phi_{\mathbf{y}}$ instead, we define
\begin{align*}
	&\begin{aligned}[c]
			&\widehat\Phi_{\mathbf{y}}(\boldsymbol{\xi},\mathbf{m})=(-\Phi_{\mathbf{y}})^\diamond(\boldsymbol{\xi},\mathbf{m})-\frac{\mu_0}{2}|\mathbf{m}|^2
		\end{aligned}\label{H0p}\tag{H0$'$}\\
\intertext{and assume that}
	&\begin{aligned}[c]
		\parbox{.84\textwidth}{
			$(\boldsymbol{\xi},\mathbf{b})\mapsto \Phi_{\mathbf{y}}(\boldsymbol{\xi},\mathbf{b})$, 
			$\mathbf{y}(\Omega)\times \R^3\to \R\cup\{+\infty\}$,\\
			is measurable in $\boldsymbol{\xi}$ for each $\mathbf{b}$
			and lower semi-continuous in $\mathbf{b}$ for a.e.~$\boldsymbol{\xi}$;
		}
	\end{aligned}	\tag{H1$'$}\label{H1p} \\
  &\begin{aligned}[c]
		&\text{$-C|\mathbf{b}|^2-C \leq \Phi_{\mathbf{y}}(\boldsymbol{\xi},\mathbf{b})$ for all $\mathbf{b}$ and a.e.\ $\boldsymbol{\xi}$, with a constant $C>0$;}
	\end{aligned}\tag{H2$'$}\label{H2p} \\
	&\begin{aligned}[c] 
		&\parbox{.84\textwidth}{
			for a.e. $\boldsymbol{\xi}\in\mathbf y(\Omega)$, $\Phi_{\mathbf{y}}(\boldsymbol{\xi},\cdot)$ is differentiable 
			and its gradient
			$\mathbf{b}\mapsto \nabla_{\mathbf{b}}\Phi_{\mathbf{y}}(\boldsymbol{\xi},\mathbf{b})$, $\R^3\to \R^3$,
			is a continuous and continuously invertible function.
		}	
	\end{aligned}\label{H3p}\tag{H3$'$}
\end{align*}
Here and throughout, \eqref{H1}--\eqref{H2} and \eqref{H1p}--\eqref{H2p} are imposed only to ensure that the integrals in $\widehat\cE_{\mathbf{y}}$ and $\cE_{\mathbf{y}}$ are well defined.
As stated, these conditions readily admit convex examples; alternatively, one may impose them for $-\widehat\Phi_{\mathbf{y}}$ or $-\Phi_{\mathbf{y}}$ when a strongly concave example is desired. \bbb

We emphasize that, under \eqref{H3} or \eqref{H3p}, the material energy densities appearing in \eqref{H0} and \eqref{H0p} can be related using the (smooth) Legendre--Fenchel transform. Since this transform is an involution (see Lemma~\ref{lem:duality}), the definitions \eqref{H0} and \eqref{H0p} are mutually consistent. \bbb
This equivalence does not require convexity; it only requires invertibility of the gradient.

\begin{theorem}[Model equivalence at equilibirum: smooth case]\label{thm:eq-smooth}
The following holds:
	\begin{itemize}
\item [(i)] Assume that \eqref{H0} and \eqref{H3} are satisfied. 
If $(\overline{\mathbf{m}},\overline{\mathbf{h}}_{\rm s})\in L^2(\mathbf y(\Omega);\mathbb R^3)\times L^2(\mathbb R^3;\mathbb R^3)$ is a constrained critical point of $\widehat{\cE}$ 
in the sense of Definition~\ref{def:crit-smooth}, 
then the field $\overline{\mathbf{b}}$ defined by \eqref{b-from-mh-smooth}
is a constrained critical point of $\cE$.

\item [(ii)] Assume that \eqref{H0p} and \eqref{H3p} are satisfied.
If $\overline{\mathbf{b}}\in L^2(\R^3;\mathbb R^3)$ is a constrained critical point of $\cE$
in the sense of Definition~\ref{def:crit-smooth},

then the pair $(\overline{\mathbf{m}},\overline{\mathbf{h}}_{\rm s})$ defined by \eqref{mh-from-b-smooth} belongs to $L^2(\mathbf y(\Omega);\mathbb R^3)\times L^2(\mathbb R^3;\mathbb R^3)$
and is a constrained critical point of $\widehat\cE_{\mathbf y}$. 
Moreover, it satisfies 
\begin{equation}\label{eq:mh-from-b-smooth}
\mathbf b_{{\rm a}}+\mu_0(\chi_{\mathbf{y}(\Omega)}\overline{\mathbf{m}}+\overline{\mathbf{h}}_{\rm s})=\overline{\mathbf{b}} \quad \text{a.e. in } \mathbb{R}^3.
\end{equation}
\end{itemize}
\end{theorem}
\begin{remark}\label{rem:EL}
Under our assumptions, the Euler-Lagrange equations \eqref{equil-mh-smooth} and \eqref{equil-b-smooth} of Definition~\ref{def:crit-smooth} are well-defined, but
\eqref{H3} and \eqref{H3p} alone are not strong enough to justify them as necessary conditions for local extremals via Proposition~\ref{prop:charact-mh}. For the latter, we also need \eqref{H1} and \eqref{H1p} as well as stronger versions of
\eqref{H2} and \eqref{H2p}, namely, that
$\nabla_{\mathbf{m}}\widehat\Phi_{\mathbf{y}}$, $\nabla_{\mathbf{b}}\Phi_{\mathbf{y}}$ and their inverse functions satisfy a linear growth condition.
\end{remark}
\bbb
\begin{proof}[Proof of Theorem~\ref{thm:eq-smooth}]
\noindent{\bf Proof of (i).} In view of the definition of $\overline{\mathbf{b}}$ in \eqref{b-from-mh-smooth}, $\overline{\mathbf{b}}\in L^2(\R^3;\R^3)$. Moreover, the first Maxwell constraint $\operatorname{div}(\chi_{\mathbf{y}(\Omega)} \overline{\mathbf{m}} + \overline{\mathbf{h}}_{\rm s}) = 0$  and the assumption that $\mathbf b_a$ is divergence-free yield immediately that $\overline{\mathbf b}$ is divergence-free:
\begin{equation}
	\operatorname{div} \overline{\mathbf{b}} = 0 \quad \text{in } \mathbb{R}^3.
\end{equation}
By the second Maxwell constraint, $\operatorname{curl} \overline{\mathbf{h}}_{\rm s} = 0$, there exists a scalar potential $\overline{\varphi} \in \mathring{W}^{1,2}(\mathbb{R}^3)$ such that
\begin{equation}
	\overline{\mathbf{h}}_{\rm s} = -\nabla \overline{\varphi}.
\end{equation}
Outside $\mathbf{y}(\Omega)$, by definition \eqref{b-from-mh-smooth}, $\overline{\mathbf{m}} = 0$ and $\overline{\mathbf{b}} = \mu_0 \overline{\mathbf{h}}_{\rm s} + \mathbf{b}_{\rm a}$. Hence,
\begin{equation}\label{eq:mh-from-b-proof1}
	\frac{\overline{\mathbf{b}} - \mathbf{b}_{\rm a}}{\mu_0} +\nabla\overline\varphi = 0 \quad \text{a.e. in } \mathbb{R}^3 \setminus \mathbf{y}(\Omega).
\end{equation}
Thus the Euler-Lagrange equation \eqref{equil-b-smooth} for the $\mathbf b$-model holds in $\mathbb{R}^3 \setminus \mathbf{y}(\Omega)$, and it remains to show that it also holds in $\mathbf{y}(\Omega)$. This is where we use the assumption that $(\overline{\mathbf{m}}, \overline{\mathbf{h}}_{\rm s})$ satisfies the Euler-Lagrange equation \eqref{equil-mh-smooth}. Indeed, substituting the definition of $\overline{\mathbf{b}}$ in \eqref{b-from-mh-smooth}, we obtain
\begin{equation}\label{eq:b-from-mh-proof}
	\overline{\mathbf{b}} = \nabla_{\mathbf{m}} \widehat{\Psi}_{\mathbf{y}}(\cdot, \overline{\mathbf{m}}) \quad \text{a.e. in } \mathbf{y}(\Omega).
\end{equation}
By the duality properties of the Legendre transform (see Lemma~\ref{lem:duality}), equation \eqref{eq:b-from-mh-proof} is equivalent to
\begin{equation}
	\overline{\mathbf{m}} = -\nabla_{\mathbf{b}} \Phi_{\mathbf{y}}(\cdot, \overline{\mathbf{b}}) \quad \text{a.e. in } \mathbf{y}(\Omega).
\end{equation}
Substituting this into the definition of $\overline{\mathbf{b}}$ in $\mathbb{R}^3$, and recalling that $\overline{\mathbf{h}}_{\rm s} = -\nabla \overline{\varphi}$, we have
\begin{equation}\label{eq:mh-from-b-proof2}
	-\chi_{\mathbf{y}(\Omega)} \nabla_{\mathbf{b}} \Phi_{\mathbf{y}}(\cdot, \overline{\mathbf{b}}) = \frac{\overline{\mathbf{b}} - \mathbf{b}_{\rm a}}{\mu_0} + \nabla \overline{\varphi} \quad \text{a.e. in } \mathbb{R}^3.
\end{equation}
The combination of \eqref{eq:mh-from-b-proof1} and \eqref{eq:mh-from-b-proof2} yields the Euler-Lagrange equation \eqref{equil-b-smooth} for $\overline{\mathbf{b}}$.

\noindent{\bf Proof of (ii).} Conversely, let $\overline{\mathbf{b}}$ be a smooth constrained critical point of $\mathcal{E}_{\mathbf{y}}$. By hypothesis, $\overline{\mathbf{b}}$ satisfies the Euler-Lagrange equation \eqref{equil-b-smooth}. By the definition of $\overline{\mathbf{m}}$ in \eqref{mh-from-b-smooth}, we obtain that
\begin{equation}\label{eq:equil}
    \frac{\overline{\mathbf{b}} - \mathbf{b}_{\rm a}}{\mu_0} + \nabla \overline{\varphi} = \chi_{y(\Omega)}\overline{\mathbf{m}} \quad \text{a.e. in } \mathbb R^3.
\end{equation}
(Note that we do not know yet whether $\overline\varphi$ is the same as in the proof of (i).)
In particular, we have that $\overline{\mathbf{m}} \in L^2$, as required by the formulation:
\begin{equation}
\overline{\mathbf{m}}\in L^2(\mathbf y(\Omega);\R^3).
\end{equation}
Applying the projection operator $\mathcal P$ to both sides of \eqref{eq:equil}, using the orthogonality beteween curl-free and divergence-free fields, and the assumption that both $\overline{\mathbf b}$ and  $\mathbf b_a$ are divergence-free, we obtain $\nabla\overline\varphi=\mathcal P[\chi_{\mathbf y(\Omega)}\overline{\mathbf m}]$, and hence the definition of $\overline{\mathbf{h}}_{\rm s}$ in \eqref{mh-from-b-smooth} yields 
\begin{equation}\label{eq:h-from-mh-proof}
	\overline{\mathbf{h}}_{\rm s}=-\nabla\overline\varphi.
\end{equation}
Substituting \eqref{eq:h-from-mh-proof} into \eqref{eq:equil}, we obtain \eqref{eq:mh-from-b-smooth}. Moreover, since both $\overline{{\mathbf{b}}}$ and $\mathbf{b}_a$ are divergence-free, we have $\operatorname{div}(\chi_{\mathbf{y}(\Omega)}\overline{\mathbf{m}}+\overline{\mathbf{h}}_{\rm s})=0$. Thus, the pair $(\overline{\mathbf{m}}, \overline{\mathbf{h}}_{\rm s})$ satisfies the Maxwell constraints.

Finally, in view of \eqref{H0} and the duality properties of the Legendre transform (see Lemma~\ref{lem:duality}), the definition \eqref{mh-from-b-smooth} of $\overline{\mathbf{m}}$ is equivalent to
\begin{equation}\label{eq:equil-mh-proof}
    \overline{\mathbf{b}} = \nabla_{\mathbf{m}} \widehat{\Psi}_{\mathbf{y}}(\cdot, \overline{\mathbf{m}}),\quad\text{a.e. in } \mathbf{y}(\Omega).
\end{equation}
Substituting \eqref{eq:h-from-mh-proof} and \eqref{eq:equil-mh-proof} into \eqref{eq:equil} (evaluated only in $\mathbf{y}(\Omega)$), we recover the Euler-Lagrange equation \eqref{equil-mh-smooth} for $(\overline{\mathbf{m}}, \overline{\mathbf{h}}_{\rm s})$.
\end{proof}

\begin{remark}
The proof of (ii) of Theorem~\ref{thm:eq-smooth}, shows that if $\overline{\mathbf b}$ is a constrained critical point then the Lagrange multiplier $-\nabla\overline\varphi$ in the Euler-Lagrange equation \eqref{equil-b-smooth} coincides with the stray field $\overline{\mathbf{h}}_{\rm s}$ that corresponds to $\overline{\mathbf b}$ through \eqref{mh-from-b-smooth}.
\end{remark}

\begin{remark}\label{rem:alternative-h-cc}
In the statement of Theorem \ref{thm:eq-smooth} (ii), the definition of $\overline{\mathbf{h}}_{\rm s}$ in \eqref{mh-from-b-smooth} is not the only possible one. In particular, the second of \eqref{mh-from-b-smooth} may be replaced by
$$
\overline{\mathbf{h}}_{\rm s}:=-\chi_{\mathbf{y}(\Omega)}\overline{\mathbf{m}}+\tfrac{1}{\mu_0}\left(\overline{\mathbf{b}}-\mathbf{b}_{\rm a}\right).
$$
This definition, along with the assumption that $\overline{\mathbf{b}}$ and $\mathbf b_{\rm a}$ are divergence-free, automatically guarantees \eqref{eq:mh-from-b-smooth}, and $\operatorname{div}(\chi_{\mathbf{y}(\Omega)}\overline{\mathbf m}+\overline{\mathbf{h}}_{\rm s})=0$, but it does not guarantee the second Maxwell constraint, namely, that $\overline{\mathbf{h}}_{\rm s}$ is curl-free, unless we assume that $\overline{\mathbf{b}}$ satisfies the Euler-Lagrange equation \eqref{equil-b-smooth}. In contrast, the definition of $\overline{\mathbf{h}}_{\rm s}$ in \eqref{mh-from-b-smooth} immediately implies that $\overline{\mathbf{h}}_{\rm s}$ is curl-free as the projection of $-\chi_{\mathbf{y}(\Omega)}\overline{\mathbf{m}}$ on the space of curl-free vector fields. However, it does not guarantee the other Maxwell constraint $\operatorname{div}(\chi_{\mathbf{y}(\Omega)}\overline{\mathbf m}+\overline{\mathbf{h}}_{\rm s})=0$ unless we assume that $\overline{\mathbf{b}}$ satisfies the Euler-Lagrange equation \eqref{equil-b-smooth}.
\end{remark}

\subsection{Critical states and energy equality}\label{ssec:critstates}
We have seen in the previous subsection that to each constrained critical point of $\widehat{\mathcal E}_{\mathbf{y}}$ corresponds a constrained critical point of $\mathcal E_{\mathbf{y}}$, and vice versa. Most importantly, the correspondence between critical points is one to one (at least in the smooth case treated in this section). This correspondence allows us to define the notion of critical state, which is a triplet $(\overline{\mathbf m},\overline{\mathbf{h}}_{\rm s},\overline{\mathbf b})$ that satisfies the Maxwell equations and the Euler-Lagrange equations of both models simultaneously. The precise definition is as follows.

\begin{definition}[Critical states]\label{def:critstate}
Under the same assumptions as in Definition \ref{def:crit-smooth} and either (i) or (ii) in  Theorem~\ref{thm:eq-smooth}, a triplet $(\overline{\mathbf m},\overline{\mathbf{h}}_{\rm s},\overline{\mathbf b})\in L^2(\mathbf{y}(\Omega); \mathbb{R}^3)\times L^2(\mathbb{R}^3; \mathbb{R}^3)\times L^2(\mathbb{R}^3; \mathbb{R}^3)$ is a critical state if
\begin{equation}\label{eq:maxdefb}
	\operatorname{curl}\overline{\mathbf{h}}_s=0, \quad \operatorname{div}\overline{\mathbf{b}}=0,\qquad \overline{\mathbf b}=\mathbf{b}_{\rm a}+\mu_0(\chi_{\mathbf{y}(\Omega)}\overline{\mathbf{m}}+\overline{\mathbf{h}}_{\rm s}), 
\end{equation}
and
\begin{equation}\label{eq:duality3}
	\nabla_{\mathbf{m}} \widehat{\Psi}_{\mathbf{y}}\left(\cdot,\overline{\mathbf{m}}\right)=\overline{\mathbf b}  \quad \text { a.e. in } \mathbf{y}(\Omega),
\end{equation}		
or, equivalently,
\begin{equation}\label{eq:duality4}
	 -\nabla_{\mathbf{b}}\Phi_{\mathbf{y}}(\cdot,\overline{\mathbf{b}})=\overline{\mathbf m} \quad \text { a.e. in } \mathbf{y}(\Omega).
\end{equation}	
\end{definition}
\noindent Here, recall that $\widehat{\Psi}_{\mathbf{y}}(\boldsymbol{\xi},\mathbf{m})=\widehat{\Phi}_{\mathbf{y}}(\boldsymbol{\xi},\mathbf{m})+\frac{\mu_0}{2}|\mathbf{m}|^2$ 
is the augmented density introduced in \eqref{def:hatPsi}.

\begin{remark}\label{rem:equivalent-def-consistent}
As the smooth Legendre-Fenchel transformation is an involution by Lemma~\ref{lem:duality}, the duality relations \eqref{H0} and 
\eqref{H0p} are equivalent. Using either one and \eqref{ldual-dual1}, we see that \eqref{eq:duality3} and \eqref{eq:duality4} are indeed equivalent, too. 
\end{remark}
\begin{remark}\label{rem:link-to-equilibirum}
The third of \eqref{eq:maxdefb}, i.e., the physical induction relation
$\overline{\mathbf{b}}=\mathbf{b}_{\rm a}
+\mu_0(\chi_{\mathbf{y}(\Omega)}\overline{\mathbf{m}}+\overline{\mathbf{h}}_{\rm s})$,
holds in both directions of Theorem~\ref{thm:eq-smooth}:
in direction~(i) it is the very definition \eqref{b-from-mh-smooth} of 
$\overline{\mathbf{b}}$ from $(\overline{\mathbf{m}},\overline{\mathbf{h}}_{\rm s})$;
in direction~(ii) it is the conclusion \eqref{eq:mh-from-b-smooth} of 
the proof once $(\overline{\mathbf{m}},\overline{\mathbf{h}}_{\rm s})$ 
is defined via \eqref{mh-from-b-smooth}.

The duality condition \eqref{eq:duality3}, i.e., 
$\nabla_{\mathbf{m}}\widehat{\Psi}_{\mathbf{y}}(\cdot,\overline{\mathbf{m}})
=\overline{\mathbf{b}}$,
is intrinsically linked to equilibrium: combined with the third of 
\eqref{eq:maxdefb}, it is equivalent to the Euler-Lagrange equation 
\eqref{equil-mh-smooth} for $(\overline{\mathbf{m}},\overline{\mathbf{h}}_{\rm s})$., while the first two of \eqref{eq:maxdefb} are just the Maxwell constraint.

In direction~(ii) of Theorem~\ref{thm:eq-smooth}, the starting assumption is that the second of \eqref{eq:maxdefb} holds, and that the Euler-Lagrange equation \eqref{equil-b-smooth} holds. The proof then shows that the equivalent 
condition \eqref{eq:duality4}, i.e., 
$-\nabla_{\mathbf{b}}\Phi_{\mathbf{y}}(\cdot,\overline{\mathbf{b}})
=\overline{\mathbf{m}}$, serves as the constitutive relation used 
to \emph{define} $\overline{\mathbf{m}}$ from $\overline{\mathbf{b}}$ 
via \eqref{mh-from-b-smooth}. The proof then shows that the pair $(\overline{\mathbf{m}},\overline{\mathbf{h}}_{\rm s})$ 
so defined automatically satisfies the Maxwell constraints and 
\eqref{equil-mh-smooth}, which is exactly \eqref{eq:duality4}.
\end{remark}

Given a critical state $(\overline{\mathbf m},\overline{\mathbf{h}}_{\rm s},\overline{\mathbf b})$, a natural question is whether or not the equality $\widehat{\cE}_{\mathbf{y}}(\overline{\mathbf{b}})=\widehat{\mathcal E}_{\mathbf{y}}(\overline{\mathbf{m}},\overline{\mathbf{h}}_{\rm s})$ holds, i.e., the energy values of the two functionals correspond. The answer to this question is positive. In particular, it hinges on the following algebraic equality which is a consequence of the duality of the Legendre-Fenchel transform. 
\begin{proposition}\label{prop:eduality2}
Assume that $\overline{\mathbf m}\in L^2(\mathbf{y}(\Omega); \mathbb{R}^3)$ and $\overline{\mathbf b}\in L^2(\mathbb{R}^3; \mathbb{R}^3)$ are related by the relation \eqref{eq:duality3}. Then 
\begin{equation}\label{eduality2}
	\widehat\Psi(\cdot,\overline{\mathbf{m}})+\widehat\Psi^\diamond(\cdot,\overline{\mathbf{b}})=\overline{\mathbf{m}}\cdot \overline{\mathbf{b}}\quad \text{on $\mathbf{y}(\Omega)$.} 
\end{equation}
\end{proposition}
\begin{proof}
This follows from Lemma~\ref{lem:duality}.
\end{proof}\bbb
The next propostion explains why \eqref{eduality2} is a key equality to obtain the equivalence of the energies at critical points.
\begin{proposition}[Energy equality]\label{prop:ee}
Assume that either \eqref{H1}-\eqref{H2} or \eqref{H1p}-\eqref{H2p} hold. Moreover, 
let $(\overline{\mathbf m},\overline{\mathbf{h}}_{\rm s})\in L^2(\mathbf{y}(\Omega); \mathbb{R}^3)\times L^2(\mathbb R^3; \mathbb{R}^3)$ and $\overline{\mathbf b}\in L^2(\mathbb R^3; \mathbb{R}^3)$ be such that
 \eqref{eduality2} holds and that
\begin{equation}\label{b0m0h0}
	\overline{\mathbf{b}}=\mu_0(\chi_{\mathbf{y}(\Omega)}\overline{\mathbf{m}}+\overline{\mathbf{h}}_{\rm s})+\mathbf{b}_{\rm a}\quad \text{on $\R^3$}.
\end{equation}
Then
\begin{equation}\label{ee}
\mathcal E_{\mathbf{y}}(\overline{\mathbf{b}})=\widehat{\mathcal E}_{\mathbf{y}}(\overline{\mathbf{m}},\overline{\mathbf{h}}_{\rm s}).
\end{equation}
\end{proposition}

\begin{corollary}[Correspondence between energy values at constrained critical points]\label{cor:eneq-smooth}
Assume that either \eqref{H0}--\eqref{H3} or \eqref{H0p}--\eqref{H3p} are satisfied. \bbb
If $(\overline{\mathbf{m}},\overline{\mathbf{h}}_{\rm s},\overline{\mathbf{b}})$ is a critical state, then the following energy equality holds:
\begin{equation}\label{eq:eneq-smooth}
	\widehat{\cE}_{\mathbf{y}}(\overline{\mathbf{m}},\overline{\mathbf{h}}_{\rm s})=\cE_{\mathbf{y}}(\overline{\mathbf{b}}).
\end{equation}
\end{corollary}

\begin{proof}[Proof of Proposition~\ref{prop:ee}]
From \eqref{b0m0h0} it follows that
\begin{equation}
	\frac{\mu_0}{2} |\overline{\mathbf{h}}_{\rm s}|^2=\frac{1}{2\mu_0} |\overline{\mathbf{b}}-\mathbf{b}_{\rm a}|^2+\frac{\mu_0}{2}\chi_{\mathbf{y}(\Omega)}|\overline{\mathbf{m}}|^2-(\overline{\mathbf{b}}-\mathbf{b}_{\rm a})\cdot \overline{\mathbf{m}}.
\end{equation}
This, along with \eqref{eduality2} and the definitions of $\widehat\cE_{\mathbf y}(\overline{\mathbf{m}},\overline{\mathbf{h}}_{\rm s})$ and $\cE_{\mathbf{y}}(\overline{\mathbf{b}})$, gives the following chain of equalities:
\begin{equation*}
	\begin{aligned}
	\widehat\cE_{\mathbf y}(\overline{\mathbf{m}},\overline{\mathbf{h}}_{\rm s})&=
			 \int_{\mathbf{y}(\Omega)} \big( \widehat\Psi_{\mathbf{y}}(\boldsymbol{\xi},\overline{\mathbf{m}})
			-\overline{\mathbf m}\cdot\overline{\mathbf b} \big) \,{\rm d}\boldsymbol{\xi}
				+\int_{\mathbb R^3}\frac{1}{2\mu_0}|\overline{\mathbf{b}}-\mathbf{b}_{\rm a}|^2 \,{\rm d}\boldsymbol{\xi}
				\\
				&=\int_{\mathbf{y}(\Omega)} -\widehat\Psi_{\mathbf{y}}^\diamond(\boldsymbol{\xi},\overline{\mathbf{b}}) \,{\rm d}\boldsymbol{\xi}
				+\int_{\mathbb R^3}\frac{1}{2\mu_0}|\overline{\mathbf{b}}-\mathbf{b}_{\rm a}|^2 \,{\rm d}\boldsymbol{\xi}\\
				&= \cE_{\mathbf{y}}(\overline{\mathbf{b}}).
	\end{aligned}
\end{equation*}
\end{proof}
\begin{proof}[Proof of Corollary~\ref{cor:eneq-smooth}]
In view of Proposition~\ref{prop:ee}, we only have to check \eqref{eduality2} and \eqref{b0m0h0}. The former is due to Proposition~\ref{prop:eduality2}, while the latter is given
by \eqref{eq:maxdefb}.
\end{proof}

\begin{remark}
To verify \eqref{ee} in the proof of Proposition~\ref{prop:ee}, we have used not only
\eqref{b0m0h0} and the duality property \eqref{eduality2}
(which holds by convex or concave duality in the convex/concave case,
and by standard properties of the Legendre-Fenchel transform in the smooth case),
but also the remarkable coincidence that their combination completely eliminates
both $\overline{\mathbf{b}}$ and $\overline{\mathbf{m}}$ in the difference
$\cE_{\mathbf{y}}(\overline{\mathbf{b}})-\widehat\cE_{\mathbf{y}}(\overline{\mathbf{m}},\overline{\mathbf{h}}_{\rm s})$. However, strictly speaking, energy equality does not require that the triplet $(\overline{\mathbf m},\overline{\mathbf{h}}_{\rm s},\overline{\mathbf b})$ is a critical state in the sense of Definition~\ref{def:critstate}. 
\end{remark}

\begin{remark}
If the functions $\widehat\Psi_{\mathbf{y}}(\boldsymbol{\xi},\cdot)$ and $\Phi_{\mathbf{y}}(\boldsymbol{\xi},\cdot)$ are convex or concave, respectively, the duality relation  \eqref{eduality2} is not only a consequence of \eqref{eq:duality4}, but it in fact an equivalent condition. Thus, in view of Remark \ref{rem:equivalent-def-consistent}, \eqref{eq:duality4} and \eqref{b0m0h0} imply that the pair $(\overline{\mathbf m},\overline{\mathbf{h}}_{\rm s})$ satisfies the Euler-Lagrange equation of the functional $\widehat\cE_{\mathbf{y}}$. This however still does not guarantee that the pair $(\overline{\mathbf m},\overline{\mathbf h}_s)$ satisfies the Maxwell constraints.
\end{remark}

\begin{remark}We emphasize that if $(\mathbf{m},\mathbf{h}_{\rm s})$ is a generic pair satisfying the Maxwell constraint, the field  $\mathbf{b}=\mu_0(\chi_{\mathbf{y}(\Omega)}\mathbf{m}+\mathbf{h}_{\rm s})+\mathbf{b}_{\rm a}$ is divergence free, but the energies $\widehat\cE_{\mathbf{y}}(\mathbf{m},\mathbf{h}_{\rm s})$ and $\cE_{\mathbf{y}}(\mathbf{b})$ do not necessarily coincide for generic triplets $(\mathbf{m},\mathbf{h}_{\rm s},\mathbf{b})$ satisfying \eqref{eq:maxdefb}, as can be easily seen in the examples of Section~\ref{sec:ex}. For those, the energy equality is a direct consequence of either \eqref{eq:duality3} or \eqref{eq:duality4} which in turn are linked to equilibrium (see Remark \ref{rem:link-to-equilibirum}). 
\end{remark}

\begin{remark}
The proof of Proposition~\ref{prop:ee} is purely algebraic: it uses only the
Fenchel identity \eqref{eduality2} and the physical induction relation \eqref{b0m0h0},
neither of which requires differentiability of the energy densities. The same argument
therefore applies verbatim in the convex/concave setting of Section~\ref{sec:cc}
(see Proposition~\ref{prop:ee-cc}).
\end{remark}

\section{The convex/concave case}\label{sec:cc}

In this section we extend the equivalence results of Section~\ref{sec:smooth} to the nonsmooth setting. The key difference is that differentiability of the energy densities is no longer assumed. Instead, we require that $\widehat\Psi_{\mathbf{y}}(\boldsymbol{\xi},\cdot)$ or $\Phi_{\mathbf{y}}(\boldsymbol{\xi},\cdot)$ be either convex or concave, and we replace gradients by convex subdifferentials (or concave superdifferentials). The Fenchel duality machinery, which in the smooth case relied on standard properties of the Legendre-Fenchel transform, is now applied in its classical convex-analytic form (see, e.g., \cite{Ro70B}).

As in Section~\ref{sec:smooth}, the deformation $\mathbf{y}:\Omega\to\mathbb R^3$ is regarded as a fixed parameter. Throughout, we again assume that $\mathbf y(\Omega)$ is measurable and that 
$\mathbf{b}_a\in L^2(\R^3;\R^3)$ with $\div \mathbf{b}_a=0$ on $\R^3$.\bbb

\medskip

We next define the notion of constrained critical point appropriate to this nonsmooth setting. In the convex (or concave) case, critical points are characterized by subdifferential (or superdifferential) inclusions that generalize the Euler-Lagrange equations of Definition~\ref{def:crit-smooth}.

\begin{remark}
To treat using a single notation both the convex and the concave case, we use the symbol $\partial$ to denote either the convex subdifferential (if it is non-empty) or the concave superdifferential (otherwise), for both functions and functionals. In the convex case, the subdifferential inclusion \eqref{equil-mh} characterizes global minimizers, while in the concave case it characterizes global maximizers. More details are given in the Appendix.
\end{remark}

\begin{remark}
We note that in more general nonsmooth settings beyond the convex/concave framework, defining meaningful notions of stationarity becomes significantly more delicate \cite{Cla90B,Dol14a,Dol21a}.
Many generalized notions of differentiability formally give rise to \emph{necessary} conditions for local extrema, but these can be substantially weaker than their smooth counterparts 
and may fail to characterize stationarity in dynamical context (for an associated gradient flow, e.g.). For convex (or concave) functionals, this is not a concern, as the subdifferential conditions characterize global minimizers (maximizers) and are thus physically meaningful. Moreover, stationarity notions for \emph{semiconvex} (i.e., convex up to a quadratic perturbation) 
or \emph{semiconcave} functionals like ours can be reduced to fully convex/concave cases and are also safe.
\end{remark}

\begin{definition}[Nonsmooth constrained critical points]\label{def:crit}
\begin{enumerate}
    \item[] \hspace{-2.5em}
    \item[(i)] We say that $(\overline{\mathbf{m}},\overline{\mathbf{h}}_{\rm s})\in L^2(\mathbf{y}(\Omega);\R^3)\times L^2(\R^3;\R^3)$ is a \emph{constrained critical point of $\widehat\cE_{\mathbf{y}}$} if $\mathbf h_s$ is curl-free, $\chi_{\mathbf{y}(\Omega)}\overline{\mathbf{m}}+\overline{\mathbf{h}}_{\rm s}$ is divergence-free, and the following subdifferential inclusion holds:
    \begin{equation}\label{equil-mh}
    	\mu_0 \overline{\mathbf{m}}+\mu_0 \overline{\mathbf{h}}_{\rm s}+ \mathbf{b}_{\rm a} \in \partial_{\mathbf{m}} \widehat{\Psi}_{\mathbf{y}}\left(\cdot,\overline{\mathbf{m}}\right) \quad \text { a.e. in } \mathbf{y}(\Omega),
    \end{equation}
    where $\widehat{\Psi}_{\mathbf{y}}$ is defined in \eqref{eduality2}. 
    \item[(ii)] A field $\overline{\mathbf b}\in L^2(\R^3;\R^3)$ is a \emph{constrained critical point of $\cE_{\mathbf{y}}$} if $\operatorname{div}\overline{\mathbf{b}}=0$ in $\R^3$, and there exists $\overline{\varphi}\in \mathring W^{1,2}(\R^3)$ such that
    \begin{equation}\label{equil-b}
    	\frac{\overline{\mathbf{b}}-\mathbf{b}_{\rm a}}{\mu_0}+\nabla\overline{\varphi}\in -\chi_{\mathbf{y}(\Omega)}\partial_{\mathbf{b}}\Phi_{\mathbf{y}}(\cdot,\overline{\mathbf{b}})\qquad\text{a.e. in }\mathbb R^3.
    \end{equation}
\end{enumerate}
\end{definition}
\noindent Here and throughout, $\partial_{\mathbf{m}}\widehat{\Psi}_{\mathbf{y}}$ denotes the convex subdifferential with respect to the variable $\mathbf{m}\in\R^3$ if 
$\widehat{\Psi}_{\mathbf{y}}(\boldsymbol{\xi},\cdot)$ is convex, and the concave superdifferential if $\widehat{\Psi}_{\mathbf{y}}(\boldsymbol{\xi},\cdot)$ is concave. 
The notation for the sub- or superdifferential of $\Phi_{\mathbf{y}}$ with respect to $\mathbf{b}\in\R^3$ is analogous.
For more details see Appendix~\ref{sssec:convexanalysis}. 

As shown in Appendix~\ref{secC:crit} for convex energy densities, the inclusions \eqref{equil-b} and \eqref{equil-mh} are necessary conditions for constrained local minimizers of $\cE_{\mathbf{y}}$ and $\widehat\cE_{\mathbf{y}}$, respectively, and sufficient conditions for global minimizers. In the convex or concave case, they generalize the Euler-Lagrange equations \eqref{equil-b-smooth} and \eqref{equil-mh-smooth} of Definition~\ref{def:crit-smooth}: for a differentiable convex (or concave) function the subdifferential reduces to the singleton consisting of the gradient. In the saddle case, however, the subdifferential inclusions \eqref{equil-b} and \eqref{equil-mh} are not available, and one must rely instead on the smooth framework of Section~\ref{sec:smooth}.

\begin{remark}\label{rem:Psi-vs-Phi-subdiff}
Since $\widehat{\Psi}_{\mathbf{y}}(\boldsymbol{\xi},\mathbf{m})=\widehat{\Phi}_{\mathbf{y}}(\boldsymbol{\xi},\mathbf{m})+\frac{\mu_0}{2}|\mathbf{m}|^2$,
\eqref{equil-mh} is formally the same as
\begin{equation}\label{equil-mh2}
	\mu_0 \overline{\mathbf{h}}_{\rm s}+ \mathbf{b}_{\rm a} \in \partial_{\mathbf{m}} \widehat{\Phi}_{\mathbf{y}}\left(\cdot,\overline{\mathbf{m}}\right).
\end{equation}
However, our examples indicate that $\widehat{\Psi}_{\mathbf{y}}$ can reasonably be expected to be always either convex or concave in $\mathbf{m}$, while
$\widehat{\Phi}_{\mathbf{y}}$ might be neither (cf.~Example~\ref{ex:saturation}). Hence, the formulation \eqref{equil-mh} using $\widehat\Psi_{\mathbf{y}}$ is preferable in the nonsmooth setting.
\end{remark}

\begin{remark}
In the smooth case, \eqref{equil-mh} and \eqref{equil-b} reduce to the Euler-Lagrange equations \eqref{equil-mh-smooth} and \eqref{equil-b-smooth} of Definition~\ref{def:crit-smooth}, since for a differentiable convex (or concave) function the subdifferential is the singleton consisting of the gradient.
\end{remark}

\subsection{Correspondence between constrained critical points}

We now establish the correspondence between constrained critical points of the two models in both directions.

\begin{theorem}[From the $(\mathbf{m},\mathbf{h}_{\rm s})$-based formulation to the $\mathbf{b}$-based formulation]\label{thm:mh-to-b}
Suppose that \eqref{H1} and \eqref{H2} hold. In addition, recall that $\widehat\Psi_{\mathbf{y}}(\boldsymbol{\xi},\mathbf{m})=\widehat\Phi_{\mathbf{y}}(\boldsymbol{\xi},\mathbf{m})+\frac{\mu_0}{2}|\mathbf{m}|^2$
and assume that
\begin{align*}
  &\begin{aligned}[t]
		&\text{either $\widehat\Psi_{\mathbf{y}}$ or $-\widehat\Psi_{\mathbf{y}}$ is proper (\emph{i.e.} not everywhere infinite),\ } \\
&\text{lower semi-continuous and convex in $\mathbf{m}$, for a.e.\ $\boldsymbol{\xi}$.}
	\end{aligned}
	\tag{H4}\label{H4}
\end{align*}
Moreover, let $\Phi_y$ be defined by \eqref{H0} and let
$(\overline{\mathbf{m}},\overline{\mathbf{h}}_{\rm s})$ be a constrained critical point of $\widehat{\cE}_{\mathbf{y}}$ in the sense of Definition~\ref{def:crit} 
which satisfy the Maxwell constraint \eqref{divcon}. Then $\overline{\mathbf b}$ defined by \eqref{b-from-mh-smooth}, i.e.,
\begin{align}\label{b-from-mh}
   \overline{\mathbf{b}}=\mathbf{b}_{\rm a}+\mu_0(\chi_{\mathbf{y}(\Omega)}\overline{\mathbf{m}}+\overline{\mathbf{h}}_{\rm s})
\end{align}
is a constrained critical point of $\cE_{\mathbf{y}}$.
\end{theorem}

\begin{proof}
As in the proof of Theorem~\ref{thm:eq-smooth}, we observe that $\overline{\mathbf{b}}$ automatically satisfies the Maxwell constraint \eqref{divcon} by construction, and we only need to verify that $\overline{\mathbf{b}}$ satisfies the subdifferential inclusion \eqref{equil-b}. Since $(\overline{\mathbf{m}},\overline{\mathbf{h}}_{\rm s})$ is a constrained critical point of $\widehat{\cE}_{\mathbf{y}}$, the inclusion \eqref{equil-mh} holds.
Substituting the definition of $\overline{\mathbf{b}}$ from \eqref{b-from-mh-smooth}, this  inclusionreads
\begin{equation}\label{cc-inclusion1}
\overline{\mathbf{b}}(\boldsymbol{\xi})\in \partial_{\mathbf{m}} \widehat{\Psi}_{\mathbf{y}}\left(\boldsymbol{\xi}, \overline{\mathbf{m}}(\boldsymbol{\xi})\right)\qquad\text{for a.e. }\boldsymbol{\xi}\text{ in }\mathbf{y}(\Omega).
\end{equation}
By standard properties of the Legendre-Fenchel dual under the convexity or concavity assumption \eqref{H4} (see \eqref{Fenchelduality}), the inclusion \eqref{cc-inclusion1} implies that
\begin{equation*}
  \overline{\mathbf{m}}(\boldsymbol{\xi})\in\partial_{\mathbf{b}}\widehat{\Psi}_{\mathbf{y}}^\diamond(\boldsymbol{\xi},\overline{\mathbf{b}}(\boldsymbol{\xi}))\quad\text{ for a.e. }\boldsymbol{\xi} \text{ in }\mathbf{y}(\Omega).
\end{equation*}
Then, by \eqref{H0},
\begin{equation}\label{cc-inclusion2}
  \overline{\mathbf{m}}(\boldsymbol{\xi})\in\partial_{\mathbf{b}}(-\Phi_{\mathbf{y}})(\boldsymbol{\xi},\overline{\mathbf{b}})\quad\text{ for a.e. }\boldsymbol{\xi} \text{ in }\mathbf{y}(\Omega).
\end{equation}
From $\operatorname{curl}\overline{\mathbf{h}}_{\rm s}=0$, there exists a scalar potential $\overline{\varphi}\in \mathring W^{1,2}(\R^3)$ such that $\overline{\mathbf{h}}_{\rm s}=-\nabla\overline{\varphi}$. From the definition \eqref{b-from-mh-smooth} we then have
\begin{equation}\label{cc-eq2222}
  	\frac{\overline{\mathbf{b}}-\mathbf{b}_{\rm a}}{\mu_0}+\nabla\overline{\varphi}=\chi_{\mathbf{y}(\Omega)}\overline{\mathbf{m}}.
\end{equation}
Combining \eqref{cc-inclusion2} and \eqref{cc-eq2222}, we obtain
\begin{equation}
  \frac{\overline{\mathbf{b}}-\mathbf{b}_{\rm a}}{\mu_0}+\nabla \overline{\varphi}\in \chi_{\mathbf{y}(\Omega)}\partial_{\mathbf{b}}(-\Phi_{\mathbf{y}})(\cdot,\overline{\mathbf{b}})\quad \text{ a.e. in }\mathbb R^3,
\end{equation}
which is exactly the inclusion \eqref{equil-b}. Hence $\overline{\mathbf{b}}$ is a constrained critical point of $\cE_{\mathbf{y}}$.
\end{proof}

\begin{theorem}[From the $\mathbf{b}$-based formulation to the $(\mathbf{m},\mathbf{h}_{\rm s})$-based formulation]\label{thm:b-to-mh}
Suppose that \eqref{H1p} and \eqref{H2p} hold. In addition, assume that
\begin{align*}
 &\text{either $\Phi_{\mathbf{y}}$ or $-\Phi_{\mathbf{y}}$ is proper, lower semi-continuous and convex in $\mathbf{b}$, for a.e.\ $\boldsymbol{\xi}$.}
	\tag{H4$'$}\label{H4p}
\end{align*}
Let $\widehat\Phi_{\mathbf{y}}$ be defined by \eqref{H0} and let
$\overline{\mathbf{b}}\in L^2(\R^3;\R^3)$ be a constrained critical point of $\cE_{\mathbf{y}}$ in the sense of Definition~\ref{def:crit} which satisfies $\div \overline{\mathbf{b}}=0$.
Then there exists $\overline{\mathbf{m}}\in L^2(\mathbf{y}(\Omega);\R^3)$ such that
\begin{equation}\label{mh-from-b}
	\overline{\mathbf{m}}\in-\partial_{\mathbf{b}} \Phi_{\mathbf{y}}\left(\cdot, \overline{\mathbf{b}}\right)\quad\text{a.e.~on }\mathbf{y}(\Omega)
\end{equation}
and the pair $(\overline{\mathbf m},\overline{\mathbf{h}}_{\rm s})$ with $\overline{\mathbf h}_s$ defined by the second of \eqref{mh-from-b-smooth}, i.e., $\overline{\mathbf{h}}_{\rm s}=-\cP[\chi_{y(\Omega)}\overline{\mathbf m}]$, is a constrained critical point of $\widehat{\cE}_{\mathbf{y}}$. In addition, \eqref{b-from-mh} holds.
\end{theorem}

\begin{proof}[Proof of Theorem~\ref{thm:b-to-mh}]
Since $\overline{\mathbf{b}}$ is a constrained critical point of $\cE_{\mathbf{y}}$, there exists $\overline{\varphi}\in \mathring W^{1,2}(\mathbb R^3)$ such that
\begin{equation}\label{cc-critical5}
	\frac{\overline{\mathbf{b}}-\mathbf{b}_{\rm a}}{\mu_0}+\nabla \overline{\varphi} \in \chi_{\mathbf{y}(\Omega)}\partial_{\mathbf{b}}(-\Phi_{\mathbf{y}})(\cdot,\overline{\mathbf{b}}) \quad\text{ a.e. in }\R^3.
\end{equation}
We define
\begin{equation}\label{cc-defm0}
	\overline{\mathbf{m}}:=\left.\left(\frac{\overline{\mathbf{b}}-\mathbf{b}_{\rm a}}{\mu_0}+\nabla \overline{\varphi}\right)\right|_{\mathbf{y}(\Omega)}\in L^2(\mathbf{y}(\Omega);\R^3).
\end{equation}
We proceed as in the proof of Theorem~\ref{thm:eq-smooth}, to show that the definition \eqref{mh-from-b-smooth}, along with \eqref{cc-defm0}, implies that 
\begin{equation}
\overline{\mathbf h_{\rm s}}=-\nabla \overline{\varphi},
\end{equation}
which, together with the hypothesis that both $\overline{\mathbf b}$ and $\mathbf b_{\rm a}$ are divergence-free, implies that the pair $(\overline{\mathbf m},\overline{\mathbf{h}}_{\rm s})$ satisfies the Maxwell constraints and that the relation \eqref{b-from-mh} holds. 

By \eqref{cc-defm0} and \eqref{cc-critical5}, we have $\overline{\mathbf{m}}\in \partial_{\mathbf{b}}(-\Phi_{\mathbf{y}})(\cdot,\overline{\mathbf{b}})$ a.e.~on $\mathbf{y}(\Omega)$.
Since $\Phi_{\mathbf{y}}$ is either convex or concave in $\mathbf{b}$ by \eqref{H4p}, we have $\Phi_{\mathbf{y}}=(\Phi_{\mathbf{y}}^\diamond)^\diamond=-\widehat\Psi_{\mathbf{y}}^\diamond$, the latter by \eqref{H0p} with $\widehat{\Psi}_{\mathbf{y}}(\boldsymbol{\xi},\mathbf{m})=\widehat{\Phi}_{\mathbf{y}}(\boldsymbol{\xi},\mathbf{m})+\frac{\mu_0}{2}|\mathbf{m}|^2$. Hence
\begin{equation}\label{cc-dual1}
	\overline{\mathbf{m}}\in \partial_{\mathbf{b}}\widehat\Psi_{\mathbf{y}}^\diamond(\cdot,\overline{\mathbf{b}})~~~\text{a.e.~on	}\mathbf{y}(\Omega).
\end{equation}
By convex (or concave) duality (see \eqref{Fenchelduality}), \eqref{cc-dual1} is equivalent to
\begin{equation}\label{cc-dual2}
	\overline{\mathbf{b}}\in\partial_{\mathbf{m}}\widehat\Psi_{\mathbf{y}}(\cdot, \overline{\mathbf{m}})~~~\text{a.e.~on	}\mathbf{y}(\Omega).
\end{equation}
Combined with \eqref{cc-defm0}, the inclusion \eqref{cc-dual2} gives the inclusion \eqref{equil-mh}. Hence $(\overline{\mathbf{m}},\overline{\mathbf{h}}_{\rm s})$ is a constrained critical point of $\widehat\cE_{\mathbf{y}}$.
\end{proof}

\begin{remark}[Alternative argument via direct global minimization]\label{rem:global-min}
Equality of energies at the global infimum
can also be shown by a direct variational argument already observed in \cite[Proposition 2.1]{PeYa09a} (see also the introduction of \cite{PedregalYan2010DualityMicromagnetics}).
It still works for our more general models and does not even require assumptions of convexity. However,
the information we get is very limited unless $\widehat \cE$ is bounded from below. Assuming such a lower bound excludes all non-constant concave $\widehat\Phi$, in particular the diamagnetic case.
The argument goes as follows.

Exploiting the variational characterisation of the projection $\cP$ (cf.~Remark~\ref{rem:h-from-m} in Appendix~\ref{secA:Helmholtz}), the stray field term
in $\widehat \cE$ can be computed by minimization of an auxiliary functional:
\begin{equation}
	\int_{\R^3} \frac{\mu_0}{2}|\mathbf{h}_{\rm s}|^2\,{\rm d}\boldsymbol{\xi}=\int_{\R^3} \frac{\mu_0}{2}|\cP[\chi_{\mathbf{y}(\Omega)}\mathbf{m}]|^2\,{\rm d}\boldsymbol{\xi}=\min_{\mathbf{b}\in L^2_{\div}} \int_{\R^3} \frac{\mu_0}{2} |\mathbf{b}-\chi_{\mathbf{y}(\Omega)}\mathbf{m}|^2 \,{\rm d}\boldsymbol{\xi},
\end{equation}
where $L^2_{\div}:=\{\mathbf{b}\in L^2(\R^3;\R^3):\div \mathbf{b}=0\}$.
Hence, for any $(\mathbf{m},\mathbf{h}_{\rm s})$ satisfying \eqref{divcon},
$\widehat\cE_{\mathbf{y}}(\mathbf{m},\mathbf{h}_{\rm s})=\min_{\mathbf{b}\in L^2_{\div}} \widehat\cF_{\mathbf{y}}(\mathbf{m},\mathbf{b})$,
where
\begin{equation}
	\widehat\cF_{\mathbf{y}}(\mathbf{m},\mathbf{b}):=\int_{\mathbf{y}(\Omega)} \big(\widehat\Phi_{\mathbf{y}}(\boldsymbol{\xi},\mathbf{m})- \mathbf{b}_{\rm a}\cdot \mathbf{m}\big)\,{\rm d}\boldsymbol{\xi}+
	\int_{\R^3}  \frac{\mu_0}{2}|\mathbf{b}-\chi_{\mathbf{y}(\Omega)}\mathbf{m}|^2\,{\rm d}\boldsymbol{\xi}.
\end{equation}
Since the order of consecutive minimizations can be exchanged,
\begin{equation}
	\inf_{(\mathbf{m},\mathbf{h}_{\rm s})\,:\,\text{\eqref{divcon} holds}}\widehat\cE_{\mathbf{y}}(\mathbf{m},\mathbf{h}_{\rm s})=\inf_{\mathbf{b}\in L^2_{\div}} \inf_{\mathbf{m}} \widehat\cF_{\mathbf{y}}(\mathbf{m},\mathbf{b}).
\end{equation}
As $\mathbf{m}\in L^2(\mathbf{y}(\Omega);\R^3)$ is unconstrained, the inner infimum on the right hand side can be computed
pointwise via the Fenchel transform of $\widehat\Psi_{\mathbf{y}}$.
This leads to $\inf_{\mathbf{m}} \widehat\cF_{\mathbf{y}}(\mathbf{m},\mathbf{b})=\cE_{\mathbf{y}}(\mathbf{b})$, with $\Phi_{\mathbf{y}}$ defined by \eqref{H0p} as before.
Consequently,
\begin{equation}
	\inf_{(\mathbf{m},\mathbf{h}_{\rm s})\,:\,\text{\eqref{divcon} holds}}\widehat\cE_{\mathbf{y}}(\mathbf{m},\mathbf{h}_{\rm s})=\inf_{\mathbf{b}\in L^2_{\div}} \cE_{\mathbf{y}}(\mathbf{b}).
\end{equation}
\end{remark}

\subsection{Energy equality}\label{ssec:ee-nonsmooth}

We now show that as in the smooth case, at any pair of corresponding critical points, the two energy functionals take the same value.

\begin{proposition}[Energy equality]\label{prop:ee-cc} Assume that either \eqref{H1}--\eqref{H2} or \eqref{H1p}--\eqref{H2p} hold,
let $\overline{\mathbf{m}}\in L^2(\mathbf{y}(\Omega);\R^3)$ and let
$\overline{\mathbf{b}},\overline{\mathbf{h}}_{\rm s}\in L^2(\R^3;\R^3)$.
In addition, suppose that the Fenchel identity
\begin{equation}\label{cc-eduality}
	\widehat\Psi_{\mathbf{y}}(\cdot,\overline{\mathbf{m}})+\widehat\Psi_{\mathbf{y}}^\diamond(\cdot,\overline{\mathbf{b}})=\overline{\mathbf{m}}\cdot \overline{\mathbf{b}}\quad \text{a.e.~on $\mathbf{y}(\Omega)$,}
\end{equation}
the physical induction relation
\begin{equation}\label{cc-b0m0h0}
	\overline{\mathbf{b}}=\mu_0(\chi_{\mathbf{y}(\Omega)}\overline{\mathbf{m}}+\overline{\mathbf{h}}_{\rm s})+\mathbf{b}_{\rm a}\quad \text{on $\R^3$,}
\end{equation}
and the density relations
\begin{equation}\label{cc-PhihatPhi-to-Psi}
	\Phi_{\mathbf{y}}(\boldsymbol{\xi},\mathbf{b})=-{\widehat\Psi_{\mathbf{y}}}^\diamond(\boldsymbol{\xi},\mathbf{b})~~~\text{and}~~~\widehat\Phi_{\mathbf{y}}(\boldsymbol{\xi},\mathbf{m})+\frac{\mu_0}{2}|\mathbf{m}|^2=\widehat\Psi_{\mathbf{y}}(\boldsymbol{\xi},\mathbf{m}),
\end{equation}
hold. Then $\cE_{\mathbf{y}}(\overline{\mathbf{b}})=\widehat{\cE}_{\mathbf{y}}(\overline{\mathbf{m}},\overline{\mathbf{h}}_{\rm s})$.
\end{proposition}

\begin{proof}
The proof is identical to that of Proposition~\ref{prop:ee}: it expands \eqref{cc-b0m0h0}
to rewrite $\frac{\mu_0}{2}|\overline{\mathbf{h}}_{\rm s}|^2$, substitutes \eqref{cc-PhihatPhi-to-Psi}
into the definitions of $\cE_{\mathbf{y}}$ and $\widehat\cE_{\mathbf{y}}$, and applies the
Fenchel identity \eqref{cc-eduality} to conclude. No smoothness of the energy densities is used.
\end{proof}

\begin{remark}
The same observations as in the remarks following Proposition~\ref{prop:ee} apply here:
the proof uses only \eqref{cc-eduality} and \eqref{cc-b0m0h0}, neither of which requires
differentiability. Verifying both conditions simultaneously requires, in practice, the equilibrium condition.
\end{remark}

\begin{corollary}[Energy equality at constrained critical points]\label{cor:eneq-cc}
Assume that either \eqref{H0}--\eqref{H2} and \eqref{H4} or \eqref{H0p}--\eqref{H2p} and \eqref{H4p} hold. If $(\overline{\mathbf{m}},\overline{\mathbf{h}}_{\rm s})$ and $\overline{\mathbf{b}}$ are related by the transition formulas \eqref{b-from-mh} or \eqref{mh-from-b}, respectively, then
\begin{equation}\label{cc-eneq}
	\widehat{\cE}_{\mathbf{y}}(\overline{\mathbf{m}},\overline{\mathbf{h}}_{\rm s})=\cE_{\mathbf{y}}(\overline{\mathbf{b}}).
\end{equation}
\end{corollary}
\begin{proof}
The proofs of Theorems~\ref{thm:mh-to-b} and \ref{thm:b-to-mh} establish
$\overline{\mathbf{b}}\in\partial_{\mathbf{m}}\widehat\Psi_{\mathbf{y}}(\cdot, \overline{\mathbf{m}})$
a.e.~on $\mathbf{y}(\Omega)$, which by convex/concave duality implies the
Fenchel identity \eqref{cc-eduality}. The physical relation \eqref{cc-b0m0h0}
and density relations \eqref{cc-PhihatPhi-to-Psi} hold by construction.
Proposition~\ref{prop:ee-cc} (equivalently, Proposition~\ref{prop:ee}) then gives the result.
\end{proof}

\section{Examples} \label{sec:ex}

We illustrate the equivalence results of the preceding sections with some representative examples of standard magnetic materials. We will only discuss the case of identity deformation $\mathbf{y}=\id$, but the same examples can be easily adapted to nontrivial deformations as well.

\subsection{Smooth cases}

\begin{example}[Diamagnetic or paramagnetic materials]\label{ex:diapara}
Paramagnetic and diamagnetic materials are described by the following linear relation between magnetic induction and magnetic field:
\begin{align}\label{cr-diapara}
	\mathbf{b}=\mu \mathbf{h}
\end{align}
The constant $\mu>0$ is the magnetic permeability of the material. The material is \emph{diamagnetic} if $\mu<\mu_0$ and \emph{paramagnetic} if $\mu>\mu_0$. Using the relation $\mathbf b=\mu_0(\mathbf m+\mathbf h)$, the equation \eqref{cr-diapara} can be equivalently rewritten as
\begin{align}\label{mag-actual-diapara}
	\mathbf{m}&=\big(\tfrac{1}{\mu_0}-\tfrac{1}{\mu}\big)\mathbf{b} 
\end{align}
\vspace*{-3ex}or 
\begin{align}	\label{mag-actual-diapara2}
	\tfrac{\mu_0}{\mu-\mu_0}\mathbf{m} &=\mathbf h.
\end{align}
Since we are in the smooth case, we can use the results from Theorem~\ref{thm:eq-smooth}.
Observe that \eqref{mag-actual-diapara} coincides with the first equation in \eqref{mh-from-b-smooth}, i.e., the explicit constitutive relation 
\begin{align}\label{mag-actual}
	\mathbf{m} &=-\nabla_{\mathbf{b}}\Phi_{\mathbf{y}}(\cdot,\mathbf{b}) \qquad\text{a.e.~in } \mathbf{y}(\Omega),
\end{align}
at the identity deformation $\mathbf{y}=\id$ if
\begin{equation}
	\Phi_{\id}(\boldsymbol{\xi},\mathbf{b})=\frac{1}{2}\Big(\frac{1}{\mu}-\frac{1}{\mu_0}\Big)|\mathbf{b}|^2.
\end{equation}
As to the $(\mathbf{m},\mathbf{h}_s)$-based model, its implicit constitutive relation given at equilibrium by the Euler-Lagrange equation \eqref{equil-mh-smooth} 
is consistent with \eqref{mag-actual-diapara2} (where $\mathbf{h}=\mathbf h_s+\frac{1}{\mu_0} \mathbf{b}_a$) for
\begin{equation}
	\widehat\Phi_{\id}(\boldsymbol{\xi},\mathbf{m})=\frac{1}{2}\frac{\mu_0^2}{\mu-\mu_0}|\mathbf m|^2
	=\frac{1}{2}\Big(\frac{1}{\mu_0}-\frac{1}{\mu}\Big)^{-1} |\mathbf{m}|^2 - \frac{\mu_0}{2}|\mathbf{m}|^2
	=(-\Phi_{\id})^\diamond(\boldsymbol{\xi},\mathbf{m})-\frac{\mu_0}{2}|\mathbf{m}|^2.
\end{equation}
At equilibrium, the inverse relation of \eqref{mh-from-b-smooth} is indeed given by \eqref{b-from-mh-smooth}, and the two densities are related by~\eqref{H0} and~\eqref{H0p}, as expected from Theorem~\ref{thm:eq-smooth}.

The energy functional of the $\mathbf{b}$-model is
\begin{equation}
	\cE_\id(\mathbf{b})=\int_\Omega \frac{1}{2}\Big(\frac{1}{\mu}-\frac{1}{\mu_0}\Big)|\mathbf{b}|^2+\frac{1}{2\mu_0}|\mathbf{b}-\mathbf{b}_a|^2\,{\rm d}\boldsymbol{\xi}
	+\int_{\R^3\setminus \Omega} \frac{1}{2\mu_0}|\mathbf{b}-\mathbf{b}_a|^2\,{\rm d}\boldsymbol{\xi}.
\end{equation}
Since $\mu>0$, it is strictly convex in $\mathbf{b}$ in both the paramagnetic and diamagnetic cases.

The energy functional of the $(\mathbf{m},\mathbf{h}_{\rm s})$-based model is
\begin{equation}
	\widehat\cE_{\id}(\mathbf{m},\mathbf{h}_{\rm s}) = \frac{1}{2}\frac{\mu_0^2}{\mu-\mu_0}\int_\Omega |\mathbf{m}|^2\,{\rm d}\boldsymbol{\xi}
	+ \frac{\mu_0}{2}\int_{\R^3} |\mathbf{h}_{\rm s}|^2\,{\rm d}\boldsymbol{\xi}
	- \int_\Omega \mathbf{m}\cdot\mathbf{b}_a\,{\rm d}\boldsymbol{\xi}.
\end{equation}
In the paramagnetic case $\mu>\mu_0$, the coefficient $\frac{\mu_0^2}{2(\mu-\mu_0)}$ is positive, so both quadratic terms are convex and $\widehat\cE_{\id}$ is jointly convex in $(\mathbf{m},\mathbf{h}_{\rm s})$.
In the diamagnetic case $\mu<\mu_0$, however, $\frac{\mu_0^2}{(\mu-\mu_0)}$ is negative, whence $\widehat\Phi_{\id}$ is strictly concave in $\mathbf{m}$. The stray-field term is still positive but not strong enough to compensate: since $\mathbf{h}_{\rm s}=-\cP[\chi_\Omega\mathbf{m}]$ by the constraints (cf.~Remark~\ref{rem:h-from-m} in Appendix~\ref{secA:Helmholtz}) and the fact that $\cP:L^2\to L^2$ is an orthogonal projection, we have $\int_{\R^3}\mathbf{h}_{\rm s}\cdot(\chi_\Omega\mathbf{m}-\mathbf{h}_{\rm s})\,{\rm d}\boldsymbol{\xi}=0$. Consequently,
\begin{equation}
	\frac{\mu_0}{2}\int_{\R^3}|\mathbf{h}_{\rm s}|^2\,{\rm d}\boldsymbol{\xi}
	= \frac{\mu_0}{2}\int_{\R^3}\big(|\chi_\Omega\mathbf{m}|^2 - |\chi_\Omega\mathbf{m}-\mathbf{h}_{\rm s}|^2\big)\,{\rm d}\boldsymbol{\xi}.
\end{equation}
Substituting back and using $\widehat\Phi_{\id}+\frac{\mu_0}{2}|\mathbf{m}|^2 = \frac{1}{2}\big(\frac{1}{\mu_0}-\frac{1}{\mu}\big)^{-1}|\mathbf{m}|^2$ yields the representation
\begin{equation}\label{hatcE-id-rep}
	\widehat\cE_\id(\mathbf{m},\mathbf{h}_{\rm s}) = \frac{1}{2}\Big(\frac{1}{\mu_0}-\frac{1}{\mu}\Big)^{-1}\int_\Omega |\mathbf{m}|^2\,{\rm d}\boldsymbol{\xi}
	- \frac{\mu_0}{2}\int_{\R^3}|\chi_\Omega\mathbf{m}-\mathbf{h}_{\rm s}|^2\,{\rm d}\boldsymbol{\xi}
	- \int_\Omega \mathbf{m}\cdot\mathbf{b}_a\,{\rm d}\boldsymbol{\xi}.
\end{equation}
For $0<\mu<\mu_0$, we still have $\frac{1}{2}(\frac{1}{\mu_0}-\frac{1}{\mu})^{-1}=\frac{\mu\mu_0}{2(\mu-\mu_0)}<0$ and the second term is non-positive. Hence, $\widehat\cE_\id$ is strictly concave and unbounded from below as $\|\mathbf{m}\|_{L^2(\Omega)}\to\infty$.
\end{example}

\begin{example}[Anisotropic mixed diamagnetic-paramagnetic materials]\label{ex:diaparamix}
As a variant of Example~\ref{ex:diapara} we can also describe a material which changes its behavior from diamagentic to
paramagnetic, depending on the direction of the applied field. While we do not know whether such behavior can be be found in nature,
it is at least conceivable that this might effectively occur in artificially designed microstructures (see, e.g., \cite{Ma-etal08a,Sea-etal10a}) or more
complex biological materials. Fix an orthonormal basis $\{\mathbf{e}_1,\mathbf{e}_2,\mathbf{e}_3\}$, and consider the constitutive equation
the constitutive relation
\begin{align}\label{cr-diaparamix}
	\mathbf{b}=\mu_p (\mathbf{h}\cdot \mathbf{e}_1)\mathbf{e}_1+\mu_d (\mathbf{h}\cdot \mathbf{e}_2)\mathbf{e}_2+\mu_d (\mathbf{h}\cdot \mathbf{e}_3)\mathbf{e}_3, \quad\text{with $0<\mu_d< \mu_0<\mu_p$},
\end{align}
where the permeabilities $\mu_p$ and $\mu_d$ are constant. The relation \eqref{cr-diaparamix} describes a material whose  behavior in the direction $\mathbf{e}_1$ and diamagnetic behavior in the directions $\mathbf{e}_2,\mathbf{e}_3$. For instance, for a paramagnetic behavior in a preferred direction $\mathbf{e}_1$, and diamagnetic behavior
in the directions $\mathbf{e}_2,\mathbf{e}_3$ .

For the single field model at the identity deformation $\mathbf{y}=\id$, \eqref{cr-diaparamix} coincides with
\eqref{mag-actual} for the saddle-type density function
\begin{equation}
	\Phi_{\id}(\boldsymbol{\xi},\mathbf{b})=\frac{1}{2}\Big(\frac{1}{\mu_p}-\frac{1}{\mu_0}\Big)(\mathbf{b}\cdot \mathbf{e}_1)^2+\frac{1}{2}\Big(\frac{1}{\mu_d}-\frac{1}{\mu_0}\Big)[(\mathbf{b}\cdot \mathbf{e}_2)^2+(\mathbf{b}\cdot \mathbf{e}_3)^2].
\end{equation}
Notice that the $\mathbf{e}_i$, representing a material anisotropy, are of course natural to prescribe in reference configuration and 
thus expected to change once a nontrivial deformation $y\neq \id$ enters the picture.
Since we are again in the smooth case, we can use \eqref{H0p} 
to compute the corresponding density
\begin{equation}
	\widehat{\Phi}_{\id}(\boldsymbol{\xi},\mathbf{m})=\frac{1}{2}\Big(\frac{1}{\mu_0}-\frac{1}{\mu_p}\Big)^{-1} (\mathbf{m}\cdot \mathbf{e}_1)^2+
	\frac{1}{2}\Big(\frac{1}{\mu_0}-\frac{1}{\mu_d}\Big)^{-1} [(\mathbf{m}\cdot \mathbf{e}_2)^2+(\mathbf{m}\cdot \mathbf{e}_3)^2]
	- \frac{\mu_0}{2}|\mathbf{m}|^2.
\end{equation} 
While $\cE_\id$ is strictly convex even here, $\widehat\cE_\id$ now has a saddle structure. 
Since $\cE_\id$ has a unique critical point (its global minimizer), so has $\widehat\cE_\id$ by Theorem~\ref{thm:eq-smooth}. 
However, the latter now is neither a local minimum nor a local maximum. 
\end{example}

\subsection{Non-smooth cases}

\begin{example}[Permanently magnetic material]\label{ex:permmag}
Let $\mathbf m_0\in L^2(\Omega;\R^3)$ be a prescribed magnetization state. For the $\mathbf b$-based model, if we set
\begin{equation}\label{Phi-permag}
	\Phi_{\id}(\boldsymbol{\xi},\mathbf{b})=-\mathbf b\cdot \mathbf{m}_0,
\end{equation}
then the constitutive equation (see \eqref{mh-from-b}):
\begin{equation}\label{mh-from-b-nonsmooth}
\mathbf m\in\partial_{\mathbf b}\Phi_{\id}(\boldsymbol{\xi},\mathbf b),
\end{equation}
becomes $\mathbf m=\mathbf m_0$ a.e.~on $\Omega$, i.e., the material is permanently magnetized with fixed magnetization $\mathbf m_0$ in its undeformed state $\mathbf{y}=\id$. 

The corresponding energy density for the $(\mathbf m,\mathbf h_s)$-based model is $\widehat\Phi(\boldsymbol{\xi},\mathbf{m})=\widehat\Psi(\boldsymbol{\xi},\mathbf{m})-\frac{\mu_0}{2}|\mathbf{m}|^2$, where $\widehat\Psi$ is given by (recall \eqref{H0}, now to be used with the nonsmooth Legendre-Fenchel transform of Definition~\ref{def:genFc}):
\begin{equation}
	\widehat\Psi(\boldsymbol{\xi},\mathbf{m})=(-\Phi)^\diamond(\boldsymbol{\xi},\mathbf m)=\begin{cases}
0 & \text{if } \mathbf{m}=\mathbf{m}_0	
\\
+\infty & \text{if } \mathbf{m}\neq \mathbf{m}_0.
\end{cases}
\end{equation}	
Indeed,
\begin{equation}
	\begin{aligned}
		(-\Phi_{\id})^\diamond(\mathbf{m}) = \sup_{\mathbf{b}} (\mathbf{m} - \mathbf{m}_0) \cdot \mathbf{b} =\begin{cases}
0 & \text{if } \mathbf{m}=\mathbf{m}_0	
\\
+\infty & \text{if } \mathbf{m}\neq \mathbf{m}_0.
\end{cases}
	\end{aligned}
\end{equation} 
Both $\cE_\id$ and $\widehat\cE_\id$ are now strictly convex, although the latter is degenerate. Here, degeneracy means that $\widehat{\cE}_{\id}$ is convex but its effective domain consists of a single point (the singleton $\{\mathbf{m}_0\}$), so that strict convexity holds trivially but the functional carries no information about variations of $\mathbf{m}$ away from $\mathbf{m}_0$. \bbb
\end{example}

\begin{remark}
Smooth approximations of the above example can be obtained if we modify \eqref{Phi-permag} with a small diamagnetic or paramagnetic perturbation, i.e., an added quadratic term. This leads to  a constitutive equation of the form
$\mathbf m=\mathbf m_{0}-\frac{\eps}{\mu_0}\mathbf b$ for a small $\eps\neq  0$.
\end{remark}

\begin{example}[Ferromagnetic materials with soft saturation]\label{ex:saturation}
A natural model for magnetic saturation, consistent with the convex framework of this paper, is to subject the magnetization to the convex constraint $|\mathbf{m}|\leq m_s$, where $m_s>0$ is the saturation intensity (typically of the order of $10^6$ A/m). This \emph{soft saturation model} arises as the large-body-limit relaxation of micromagnetics \cite{DeSimone1993EnergyminimisersLargeFerromagnetic}: in the no-exchange limit, the hard saturation constraint $|\mathbf{m}|=m_s$ is replaced by $|\mathbf{m}|\leq m_s$ upon relaxation, and the resulting convex variational problem correctly predicts the observed magnetization curves of soft ferromagnets, exhibiting linear response at low fields followed by saturation \cite{DeSimone1993EnergyminimisersLargeFerromagnetic}. We define the energy density as the indicator function of the closed ball:
\begin{equation}
	\widehat{\Phi}_{\id}^{\rm (c)}(\boldsymbol{\xi},\mathbf{m}) :=
	\begin{cases}
		0 & \text{if } |\mathbf{m}| \le m_s,\\
		+\infty & \text{otherwise.}
	\end{cases}
\end{equation}
The corresponding augmented density is
\begin{equation}
	\widehat{\Psi}^{(\rm s)}_{\id}(\boldsymbol{\xi},\mathbf{m}) = \begin{cases}
		\frac{\mu_0}{2}|\mathbf m|^2 & \text{if } |\mathbf{m}|\le m_s, \\
		+\infty & \text{otherwise.}
	\end{cases}
\end{equation}
Applying \eqref{H0} to $\widehat{\Psi}_{\id}^{(\rm s)}$
we get
\begin{equation}
	\Phi_{\id}^{\rm (s)}(\boldsymbol{\xi},\mathbf{b}) = -(\widehat{\Psi}_{\id}^{(\rm s)})^\diamond(\boldsymbol{\xi},\mathbf{b})=-(\widehat{\Psi}_{\id}^{(\rm s)})^*(\boldsymbol{\xi},\mathbf{b}),
\end{equation}
with the Legendre-Fenchel transform given by Definition~\ref{def:genFc} (in all cases, one can use \eqref{def:LeFe-convex} instead, which is also defined for $\widehat{\Psi}_{\id}^{\rm (sat)}$). \bbb Computed explicitly, this yields that
\begin{align}\label{cr-sat}
	\Phi_{\id}^{\rm (s)}(\boldsymbol{\xi},\mathbf{b})
	&= -\sup_{|\mathbf{m}|\leq m_s}\Bigl(\mathbf{b}\cdot\mathbf{m} - \tfrac{\mu_0}{2}|\mathbf{m}|^2\Bigr) 
	= \begin{cases}
	  -\frac{1}{2\mu_0}|\mathbf{b}|^2 &\text{if $|\mathbf{b}|\leq \mu_0 m_s$}\bbb\\
			-|\mathbf{b}| m_s+\tfrac{\mu_0}{2}m_s^2 &\text{else.}
		\end{cases}
\end{align}
\end{example}
\begin{remark}\label{rem:langevin}
The soft saturation model $\widehat{\Phi}_{\id}^{(\rm s)}$ arises
naturally as the zero-temperature limit of the Langevin model,
a classical model from statistical mechanics in which the energy density is
\begin{equation}\label{eq:Langevin}
    \widehat{\Phi}_{\id}^{(\rm Lan)}(\mathbf{m})
    := \kappa\left[
        \frac{|\mathbf{m}|}{m_s}
        \operatorname{arctanh}\frac{|\mathbf{m}|}{m_s}
        - \frac{1}{2}\ln\!\left(1-\frac{|\mathbf{m}|^2}{m_s^2}\right)
    \right],
    \qquad |\mathbf{m}|<m_s,
\end{equation}
with $\kappa = k_{\rm B}T/\mu_0$ proportional to the absolute temperature $T$.
This density is strictly convex on the open ball $|\mathbf{m}|<m_s$,
with gradient
\[
\nabla_{\mathbf{m}}\widehat{\Phi}_{\id}^{(\rm Lan)}(\mathbf{m})
=\frac{\kappa}{m_s}\operatorname{arctanh}\frac{|\mathbf{m}|}{m_s}\,
\frac{\mathbf{m}}{|\mathbf{m}|},
\]
which is a bijection from the open ball onto $\mathbb{R}^3$; in
particular, assumption~\eqref{H3} is satisfied, and the Langevin model
falls within the smooth framework of Section~\ref{sec:smooth}.
As $T\to 0^+$ (i.e., $\kappa\to 0^+$), the energy
$\widehat{\Phi}_{\id}^{(\rm Lan)}$ converges pointwise to the indicator
of the closed ball, recovering the soft saturation model
$\widehat{\Phi}_{\id}^{(\rm s)}$.
In the opposite limit $T\to T_c^-$ (approaching the Curie temperature),
$m_s(T)\to 0$ and the Langevin energy reduces to the quadratic model
\begin{equation}
    \widehat{\Phi}_{\id}^{(\rm Lan)}(\mathbf{m})
    \approx \frac{\kappa}{2m_s^2}|\mathbf{m}|^2,
\end{equation}
which is precisely the linear paramagnetic energy density of
Example~\ref{ex:diapara}, with susceptibility $\chi=\mu_0 m_s^2/\kappa$
satisfying the Curie law $\chi\propto 1/T$.
The Langevin model and its variants are widely used in the modeling of
magnetorheological elastomers; see, e.g.,
\cite{KankanalaTriantafyllidis2004MRE,Danas2017EffectiveResponse}.
\end{remark}

\begin{remark}[Non-convex energies and the $|\mathbf{m}|=m_s$ constraint]
\label{rem:saturation-nonconvex} 
The standard micromagnetic model with the saturation constraint $|\mathbf{m}|=m_s$ \cite{Brown1966MagnetoelasticInteractions} would correspond to the energy density 
\begin{equation}
	\widehat{\Phi}_{\id}^{(\rm sat)}(\boldsymbol{\xi},\mathbf{m}) :=
	\begin{cases}
		0 & \text{if } |\mathbf{m}|  = m_s, \\
		+\infty & \text{otherwise.}
	\end{cases}
\end{equation}
This case is not covered
by the treatment of this paper. Indeed, $\widehat{\Phi}_{\id}^{(\rm sat)}$
is neither convex nor concave, and its augmented density
$\widehat{\Psi}_{\id}^{(\rm sat)} = \widehat{\Phi}_{\id}^{(\rm sat)}
+ \frac{\mu_0}{2}|\mathbf{m}|^2$ likewise fails to be convex or concave,
so that assumption~\eqref{H4} in Theorem~\ref{thm:mh-to-b} is violated.
In fact, the convex hull of $\widehat{\Phi}_{\id}^{(\rm sat)}$ is the soft saturation model $\widehat{\Phi}^{(\rm s)}_{\id}$. 

We may formally apply
\eqref{H0} to $\widehat{\Psi}_{\id}^{(\rm sat)}$, and then we would get
\begin{equation}
	{\Phi}_{\id}^{\rm (sat)}(\boldsymbol{\xi},\mathbf{b}) =-(\widehat{\Psi}_{\id}^{(\rm sat)})^*(\boldsymbol{\xi},\mathbf{b}),
\end{equation}
where we are using \eqref{def:LeFe-convex}. Computed explicitly, this yields that
\begin{align}\label{cr-sat2}
	{\Phi}_{\id}^{\rm (sat)}(\boldsymbol{\xi},\mathbf{b}) &=
	-\sup_{|\mathbf{m}|=m_s}\Bigl(\mathbf{b}\cdot\mathbf{m} - \tfrac{\mu_0}{2}|\mathbf{m}|^2\Bigr) 
	=-|\mathbf{b}| m_s+\tfrac{\mu_0}{2}m_s^2.
\end{align}
If we now try to  use \eqref{H0p}, to recover the $(\mathbf{m},\mathbf{h}_s)$-based model from ${\Phi}_{\id}^{\rm (sat)}$ (which is concave), the Legendre--Fenchel transform of $-\Phi_{\id}^{\rm (sat)}$ is
\begin{equation}
	(-{\Phi}_{\id}^{(\rm sat)})^\diamond(\boldsymbol{\xi},\mathbf{m})
	= \sup_{\mathbf{b}}\bigl(\mathbf{m}\cdot\mathbf{b} - m_s|\mathbf{b}|\bigr) + \tfrac{\mu_0}{2}m_s^2
	= \begin{cases}
		\tfrac{\mu_0}{2}m_s^2 & \text{if } |\mathbf{m}| \leq m_s,\\
		+\infty & \text{if } |\mathbf{m}| > m_s.
	\end{cases}
\end{equation}
Therefore,  by \eqref{H0p} we now obtain the new density
\begin{equation}
	\widehat{\Phi}^{\rm (sat')}_{\id}(\boldsymbol{\xi},\mathbf{m})
	= (-{\Phi}_{\id}^{\rm (sat)})^\diamond(\boldsymbol{\xi},\mathbf{m}) - \tfrac{\mu_0}{2}|\mathbf{m}|^2
	= \begin{cases}
		\tfrac{\mu_0}{2}(m_s^2 - |\mathbf{m}|^2) & \text{if } |\mathbf{m}| \leq m_s,\\
		+\infty & \text{if } |\mathbf{m}| > m_s.
	\end{cases}
\end{equation}
This matches neither the original $\widehat{\Phi}^{\rm (sat)}_{\id}$ nor its convex hull $\widehat{\Phi}^{\rm (c)}_{\id}$. Thus, while the convex Legendre-Fenchel transform can be always applied here, it is only an involution for convex functions. We still do not even recover the 
soft saturation model $\widehat{\Phi}^{(\rm s)}_{\id}$ from $\tilde{\Phi}_{\id}$ because the transition from $\widehat{\Phi}^{(\rm s)}_{\id}$ to $\widehat{\Psi}^{(\rm s)}_{\id}$ adds a quadratic term;
in fact, $\widehat{\Psi}^{(\rm s)}_{\id}$ is not the convex hull of $\widehat{\Psi}_{\id}^{\rm (sat)}$. 

Finally, we remark that as in our other examples, $\cE_\id$ with the density $\Phi_{\id}$ is strictly convex and its minimisation problem is well-posed. This is no longer the case if 
we use ${\Phi}_{\id}^{\rm (sat)}$ instead, though. However, the relaxation of $\cE_\id$ in that case leads back to the density ${\Phi}_{\id}^{\rm (s)}$: 
the convex hull of the full density ${\Phi}_{\id}^{\rm (sat)}(\cdot,\mathbf{b})+\frac{1}{2\mu_0}|\mathbf{b}-\mathbf{b}_{\rm a}|^2$ is exactly given by ${\Phi}_{\id}^{\rm (s)}(\boldsymbol{\xi},\mathbf{b})+\frac{1}{2\mu_0}|\mathbf{b}-\mathbf{b}_{\rm a}|^2$.
\end{remark}

\begin{remark}
Generally, our theoretical framework for nonsmooth cases strictly relies on the assumption of convex (or concave) $\Phi_{y}$ or 
 the augmented material energy 
$\widehat\Psi_{y}$ 
Full equivalence of micromagnetic models (without a regularizing exchange energy \cite{Brown1966MagnetoelasticInteractions}) that enforce the strong constraint $|\mathbf{m}|=m_s$ actually cannot be expected. In fact, due to the presence of the convex stray field term, it is also not obvious that replacing $\widehat{\Phi}_{\id}^{(\rm sat)}$ by $\widehat{\Phi}^{(\rm s)}_{\id}$ gives the correct relaxation of the functional $\widehat\cE_\id$. The deeper discussion of \cite{PedregalYan2010DualityMicromagnetics} is needed there. The $\mathbf{b}$-based model using $\cE_\id$ with $\Phi_{\id}$ does not have this issue but implicitly relaxes the  saturation constraint. 
\end{remark}

\appendix

\section{Auxiliary results}
In this appendix we collect some auxiliary results that are used in the main text. They are mostly well known but provided here for the reader's convenience. In particular, we recall the $L^2$ Helmholtz decomposition on $\mathbb{R}^3$, which is a key tool in our analysis.

\subsection{Helmholtz decomposition}\label{secA:Helmholtz} 

Here and throughout the paper, we use the quotient space
\begin{equation}
	\mathring W^{1,2}(\R^d)=\left\{[u]:=\{u+\lambda:\lambda \in\R\}\,\left|~
	\begin{aligned}
		&\text{$u\in L^2_{\rm loc}(\R^d)$ and}\\ 
		&\text{$\nabla u$ exists weakly in $L^2(\R^d;\R^{d})$}
	\end{aligned}
	\right.\right\}.
\end{equation} 
It is a Hilbert space with the norm $\|[u]\|_{\mathring W^{1,2}}=(\int_{\R^d} |\nabla u|^2)^{\frac{1}{2}}$, the $W^{1,2}$-gradient semi-norm. 
As this should be clear from the context, we do not distinguish equivalence classes and their representatives in $\mathring W^{1,2}(\R^d)$. In particular, we write
$u\in \mathring W^{1,2}(\R^d)$ instead of $[u]\in \mathring W^{1,2}(\R^d)$. An important property is that if $d>2$,
for any $u\in \mathring W^{1,2}(\R^d)$ there exists a constant $u_0\in\R$ such that $u-u_0$ can be approximated with respect to the $W^{1,2}$-gradient semi-norm by smooth functions with compact support,
see \cite[Thms.~II.6.3 and II.7.1]{Ga11B}.
In this sense, $C_0^\infty(\R^d)$ is dense in 
$\mathring W^{1,2}(\R^d)$ for $d>2$, in particular for our case here where $d=3$.

\begin{lemma}[$L^2$ Helmholtz decomposition]\label{lem:Helmholtz}
Every $\mathbf{u}\in L^2(\mathbb{R}^3;\mathbb{R}^3)$ admits a unique orthogonal decomposition
\begin{equation}
	\mathbf{u}=\nabla\phi+\mathbf{w},
\end{equation}
where $\phi\in\mathring{W}^{1,2}(\mathbb{R}^3)$ and $\operatorname{div}\mathbf{w}=0$ in the distributional sense on $\mathbb{R}^3$.
The two components are orthogonal in $L^2(\mathbb{R}^3;\mathbb{R}^3)$:
\begin{equation}
	\langle\nabla\phi,\mathbf{w}\rangle_{L^2(\mathbb{R}^3)}=0.
\end{equation}
The map $\mathcal{P}:=\nabla\Delta^{-1}\operatorname{div}:L^2(\mathbb{R}^3;\mathbb{R}^3)\to L^2(\mathbb{R}^3;\mathbb{R}^3)$, 
where the Laplacian $\Delta:\mathring{W}^{1,2}(\mathbb{R}^3)\to(\mathring{W}^{1,2}(\mathbb{R}^3))^*$ acts as an isomorphism, is linear and bounded and has the following additional properties:
\begin{align*}
	&\text{$\mathcal{P}^2=\mathcal{P}$ (projection);}\\
	&\text{$\mathcal{P}^*=\mathcal{P}$ (self-adjoint in $L^2(\mathbb{R}^3;\mathbb{R}^3)$);}\\
	&\text{$\operatorname{Range}(\mathcal{P})=\{\nabla\phi:\phi\in\mathring{W}^{1,2}(\mathbb{R}^3)\}$;}\\
	&\text{$\ker(\mathcal{P})=L^2_{\operatorname{div}}(\mathbb{R}^3;\mathbb{R}^3):=\{\mathbf{w}\in L^2(\mathbb{R}^3;\mathbb{R}^3):\operatorname{div}\mathbf{w}=0\}$.}
\end{align*}
\end{lemma}

\begin{proof} This is well-known; we provide a short proof for the convenience of the reader.
The Laplacian $\Delta:\mathring{W}^{1,2}(\mathbb{R}^3)\to(\mathring{W}^{1,2}(\mathbb{R}^3))^*$ is an isomorphism by the Lax--Milgram theorem, since the norm $\|[u]\|_{\mathring{W}^{1,2}}^2=\int_{\mathbb{R}^3}|\nabla u|^2\,{\rm d}\mathbf{x}$ makes $\mathring{W}^{1,2}(\mathbb{R}^3)$ a Hilbert space and the bilinear form $([u],[v])\mapsto\int_{\mathbb{R}^3}\nabla u\cdot\nabla v\,{\rm d}\mathbf{x}$ is continuous and coercive on it.

Given $\mathbf{u}\in L^2(\mathbb{R}^3;\mathbb{R}^3)$, the functional $[v]\mapsto\int_{\mathbb{R}^3}\mathbf{u}\cdot\nabla v\,{\rm d}\mathbf{x}$ is bounded on $\mathring{W}^{1,2}(\mathbb{R}^3)$ (by Cauchy--Schwarz), hence by the Riesz theorem there exists a unique $[\phi]\in\mathring{W}^{1,2}(\mathbb{R}^3)$ such that $\int_{\mathbb{R}^3}\nabla\phi\cdot\nabla v\,{\rm d}\mathbf{x}=\int_{\mathbb{R}^3}\mathbf{u}\cdot\nabla v\,{\rm d}\mathbf{x}$ for all $[v]\in\mathring{W}^{1,2}(\mathbb{R}^3)$, i.e., $\Delta\phi=\operatorname{div}\mathbf{u}$ in $(\mathring{W}^{1,2}(\mathbb{R}^3))^*$.
Setting $\mathbf{w}:=\mathbf{u}-\nabla\phi$, we have $\operatorname{div}\mathbf{w}=\operatorname{div}\mathbf{u}-\Delta\phi=0$ distributionally.

The orthogonality $\langle\nabla\phi,\mathbf{w}\rangle_{L^2}=0$ follows by integration by parts on $\mathbb{R}^3$ (since $C_0^\infty(\R^3)$ is dense $\mathring{W}^{1,2}$, the integration by parts can be effectively carried out on a bounded domain without boundary terms appearing): 
\begin{equation}
	\int_{\mathbb{R}^3}\nabla\phi\cdot\mathbf{w}\,{\rm d}\mathbf{x}=-\int_{\mathbb{R}^3}\phi\,\operatorname{div}\mathbf{w}\,{\rm d}\mathbf{x}=0.
\end{equation}
Uniqueness of the decomposition follows from the uniqueness of $[\phi]$ as the solution of $\Delta\phi=\operatorname{div}\mathbf{u}$.

Self-adjointness of $\mathcal{P}$: for any $\mathbf{u},\mathbf{v}\in L^2(\mathbb{R}^3;\mathbb{R}^3)$, write $\mathbf{u}=\nabla\phi_u+\mathbf{w}_u$ and $\mathbf{v}=\nabla\phi_v+\mathbf{w}_v$. Then, using orthogonality of the decomposition,
\begin{equation}
	\langle\mathcal{P}\mathbf{u},\mathbf{v}\rangle_{L^2}=\langle\nabla\phi_u,\nabla\phi_v+\mathbf{w}_v\rangle_{L^2}=\langle\nabla\phi_u,\nabla\phi_v\rangle_{L^2}=\langle\mathbf{u},\nabla\phi_v\rangle_{L^2}=\langle\mathbf{u},\mathcal{P}\mathbf{v}\rangle_{L^2}.
\end{equation}
The projection property $\mathcal{P}^2=\mathcal{P}$ is immediate since $\nabla\phi_u$ is already a gradient field.
\end{proof}

\begin{remark}
Given $\mathbf{m}\in L^2(\mathbf{y}(\Omega);\mathbb{R}^3)$, Lemma~\ref{lem:Helmholtz} applied to $\mathbf{u}=\chi_{\mathbf{y}(\Omega)}\mathbf{m}$ shows that there is a unique $[\phi]\in\mathring{W}^{1,2}(\mathbb{R}^3)$ such that $\Delta\phi=-\operatorname{div}(\chi_{\mathbf{y}(\Omega)}\mathbf{m})$ in $(\mathring{W}^{1,2}(\mathbb{R}^3))^*$.
Setting $\mathbf{h}_{\rm s}:=\nabla\phi$, we have $\operatorname{curl}\mathbf{h}_{\rm s}=0$ and $\operatorname{div}(\mathbf{h}_{\rm s}+\chi_{\mathbf{y}(\Omega)}\mathbf{m})=0$, i.e., the Maxwell constraint \eqref{divcon} is satisfied.
Moreover, since $\phi=-\Delta^{-1}\operatorname{div}(\chi_{\mathbf{y}(\Omega)}\mathbf{m})$, we have $\mathbf{h}_{\rm s}=\nabla\phi=-\nabla\Delta^{-1}\operatorname{div}(\chi_{\mathbf{y}(\Omega)}\mathbf{m})=-\mathcal{P}[\chi_{\mathbf{y}(\Omega)}\mathbf{m}]$, confirming the formula of Remark~\ref{rem:h-from-m}.
\end{remark}

\begin{remark}\label{rem:h-from-m}
Perturbation of $\mathbf{h}_{\rm s}$ for fixed $\mathbf{m}$ is effectively impossible under the constraints \eqref{divcon}, as these fully determine $\mathbf{h}_{\rm s}$ as a function of $\mathbf{m}$:
As a curl-free field, $\mathbf{h}\in L^2(\R^3;\R^3)$ has a potential
$\eta\in \mathring W^{1,2}(\R^3)$ 
so that $\mathbf{h}_{\rm s}=\nabla \eta$, and the other constraint implies that
$\Delta \eta=-\operatorname{div}(\chi_{\mathbf{y}(\Omega)} \mathbf{m})$ in the weak sense in $(\mathring W^{1,2})^*$, so that
\begin{align}\label{h-as-projection}
    \mathbf{h}_{\rm s}={\mathbf{h}_{\rm s}}_m:=-\cP[\chi_{\mathbf{y}(\Omega)}\mathbf{m}],&~~~~\text{where}~
		\cP:=\nabla \Delta^{-1} \operatorname{div}, ~\cP:L^2(\R^3;\R^3)\to L^2(\R^3;\R^3). 
\end{align}
\end{remark}

\subsection{Properties of the Legendre-Fenchel transform}\label{secB:duality}

\subsubsection{Smooth case: Functions with invertible gradient}

\begin{definition}[Smooth Legendre-Fenchel transform]\label{def:genFc-smooth}
Let $f:\R^3\to\R$ be continuously differentiable and suppose that
$\nabla f:\R^3\to\R^3$ is continuous and invertible.
The \emph{smooth Legendre-Fenchel transform} of $f$ is defined by
\begin{align}\label{genFc-smooth}
    f^\diamond(\mathbf{z}^*) := \mathbf{z}^* \cdot (\nabla f)^{-1}(\mathbf{z}^*) - f\bigl((\nabla f)^{-1}(\mathbf{z}^*)\bigr)
    \quad \text{for all } \mathbf{z}^* \in \R^3.
\end{align}
This definition requires only the invertibility of $\nabla f$ and applies whether $f$ is convex, concave, or neither
(e.g., $f(\mathbf{z}) := z_1^2 + z_2^2 - z_3^2 \bbb$ for $\mathbf{z}=(z_1,z_2,z_3)\in\R^3$).
\end{definition}

\begin{remark}
When $\diamond$ is applied to a function $f(\boldsymbol{\xi},\cdot)$ that also depends on a parameter
$\boldsymbol{\xi}\in\mathbf{y}(\Omega)$, it is understood that the transform is taken with respect to
the second argument, with $\boldsymbol{\xi}$ fixed.
\end{remark}

\begin{lemma}[Properties of the Legendre-Fenchel transform for smooth functions]\label{lem:duality}
Let $f:\R^3\to \R$, $\mathbf{z}\mapsto f(\mathbf{z})$ be continuously differentiable. In addition, assume that 
$\nabla f:\R^3\to \R^3$
is invertible with a globally Lipschitz-continuous inverse $(\nabla f)^{-1}$.
Then the function $f^\diamond$ defined as in \eqref{genFc-smooth}, i.e.,
\begin{align}\label{genFc-smooth2}
	f^\diamond(\mathbf{z}^*)=\mathbf{z}^*\cdot (\nabla f)^{-1}(\mathbf{z}^*)-f\bigl((\nabla f)^{-1}(\mathbf{z}^*)\bigr)~~~\text{for}~\mathbf{z}^*\in \R^3,
\end{align}
is also differentiable and its derivative satisfies 
\begin{align}\label{ldual-nabla}
	\nabla f^\diamond=(\nabla f)^{-1}~~\text{on}~\R^3.
\end{align}
In particular, 
we have the linear growth condition
\begin{align}\label{ldual-lingrowth}
	|\nabla f^\diamond(\mathbf{z}^*)|\leq \|\nabla (\nabla f)^{-1}\|_{L^\infty(\R^3;\R^{3\times 3})} |\mathbf{z}^*|+|(\nabla f)^{-1}(\mathbf{0})|~~\text{for all }\mathbf{z}^*\in\R^3.
\end{align}
Moreover, 
for any pair $(\mathbf{z},\mathbf{z}^*)\in \R^3\times\R^3$, 
\begin{align}\label{ldual-dual1}
	\mathbf{z}=\nabla f^\diamond(\mathbf{z}^*)\quad\text{if and only if}\quad \mathbf{z}^*=\nabla f(\mathbf{z}),
\end{align}
any of the two equations in \eqref{ldual-dual1} implies the Fenchel identity
\begin{align}\label{ldual-dual2}
	f^\diamond(\mathbf{z}^*)+f(\mathbf{z})=\mathbf{z}^*\cdot \mathbf{z},
\end{align}
and $(f^\diamond)^\diamond(\mathbf{z})=f(\mathbf{z})$ for all $\mathbf{z}\in\R^3$, i.e., the smooth Legendre-Fenchel transform is an involution.
\end{lemma}
\begin{proof}
We adopt the convention that $\mathbf{z}$, $\nabla f^\diamond(\mathbf{z}^*)$ and $(\nabla f)^{-1}(\mathbf{z}^*)$ are column vectors, while $\mathbf{z}^*$ and $\nabla f(\mathbf{z})$ are row vectors.

\emph{Proof of \eqref{ldual-nabla}.}
Differentiating the definition \eqref{genFc-smooth2} by the chain rule, if $(\nabla f)^{-1}$ is of class $C^1$,
\begin{align}
	\nabla f^\diamond(\mathbf{z}^*)=(\nabla f)^{-1}(\mathbf{z}^*)+\big[\mathbf{z}^*-(\nabla f)((\nabla f)^{-1}(\mathbf{z}^*))\big](\nabla (\nabla f)^{-1}(\mathbf{z}^*))=(\nabla f)^{-1}(\mathbf{z}^*),
\end{align}
since the term in the square brackets vanishes by definition of $(\nabla f)^{-1}$.
Using difference quotients instead of the chain rule, the same conclusion holds if $(\nabla f)^{-1}$ is only Lipschitz continuous (and hence differentiable almost everywhere by Rademacher's theorem), since the Lipschitz constant provides a uniform bound for its difference quotients.
This proves \eqref{ldual-nabla}.
The bound \eqref{ldual-lingrowth} is a straightforward consequence of \eqref{ldual-nabla} and the Lipschitz property of $(\nabla f)^{-1}$.

\emph{Proof of \eqref{ldual-dual1}.}
By \eqref{ldual-nabla} and the invertibility of $\nabla f$, we have
\begin{equation}
	\mathbf{z}=\nabla f^\diamond(\mathbf{z}^*)=(\nabla f)^{-1}(\mathbf{z}^*)\quad\iff\quad \mathbf{z}^*=\nabla f(\mathbf{z}),
\end{equation}
which is exactly \eqref{ldual-dual1}.

\emph{Proof of \eqref{ldual-dual2}.}
Assuming either equation in \eqref{ldual-dual1} holds, i.e., $\mathbf{z}=(\nabla f)^{-1}(\mathbf{z}^*)$, we substitute directly into \eqref{genFc-smooth2}:
\begin{equation}
	f^\diamond(\mathbf{z}^*)=\mathbf{z}^*\cdot(\nabla f)^{-1}(\mathbf{z}^*)-f((\nabla f)^{-1}(\mathbf{z}^*))=\mathbf{z}^*\cdot\mathbf{z}-f(\mathbf{z}),
\end{equation}
which gives \eqref{ldual-dual2}.
\end{proof}

\emph{Proof of the involution.}
Since $\nabla f^\diamond=(\nabla f)^{-1}$ by \eqref{ldual-nabla}, applying Definition~\ref{def:genFc-smooth} to $f^\diamond$ (whose gradient $(\nabla f)^{-1}$ has inverse $\nabla f$) gives, for any $\mathbf{w}\in\R^3$,
\begin{equation}
	(f^\diamond)^\diamond(\mathbf{w})=\mathbf{w}\cdot((\nabla f)^{-1})^{-1}(\mathbf{w})-f^\diamond\!\left(((\nabla f)^{-1})^{-1}(\mathbf{w})\right)=\mathbf{w}\cdot\nabla f(\mathbf{w})-f^\diamond(\nabla f(\mathbf{w})).
\end{equation}
Setting $\mathbf{z}=\mathbf{w}$ and $\mathbf{z}^*=\nabla f(\mathbf{w})$ in \eqref{ldual-dual2} yields $f^\diamond(\nabla f(\mathbf{w}))=\nabla f(\mathbf{w})\cdot\mathbf{w}-f(\mathbf{w})$. Substituting back gives $(f^\diamond)^\diamond(\mathbf{w})=f(\mathbf{w})$.

\subsubsection{Nonsmooth case: Convex or concave functions}\label{sssec:convexanalysis}

For completeness, we start with a recapitulation of the Fenchel-Legendre transform and relevant notions from convex analysis (see, e.g., \cite{Ro70B}). For a proper (i.e., not infinite everywhere) function $f:\R^d\to\R\cup\{+\infty\}$, its (convex) \emph{Fenchel-Legendre conjugate} is
\begin{equation}\label{def:LeFe-convex}
	f^*(\mathbf{z}^*) := \sup_{\mathbf{z}\in\R^d} ( \mathbf{z}^* \cdot \mathbf{z} - f(\mathbf{z}) ),
\end{equation}
and its \emph{convex subdifferential} at $\mathbf{z}$ is
\begin{equation}
	\partial f(\mathbf{z}):=\bigl\{\mathbf{z}^*\in\R^d : f(\mathbf{w})\geq f(\mathbf{z})+\mathbf{z}^*\cdot(\mathbf{w}-\mathbf{z})\ \forall\,\mathbf{w}\in\R^d\bigr\}.
\end{equation}
The \emph{Fenchel duality relation} states
\begin{equation}
	\mathbf{z}^*\in\partial f(\mathbf{z})\quad\iff\quad f(\mathbf{z})+f^*(\mathbf{z}^*)=\mathbf{z}\cdot\mathbf{z}^*.
\end{equation} 
Moreover, if $f$ is convex, proper and lower semicontinuous, then $f=f^{**}:=(f^*)^*$
and we even have that
\begin{equation}\label{Fenchelduality}
	\mathbf{z}^*\in\partial f(\mathbf{z})\quad\iff\quad f(\mathbf{z})+f^*(\mathbf{z}^*)=\mathbf{z}\cdot\mathbf{z}^*\quad\iff\quad \mathbf{z} \in\partial f^*(\mathbf{z}^*).
\end{equation} 
For a proper 
function $f:\R^d\to\R\cup\{-\infty\}$,\bbb 
the \emph{concave conjugate} is defined as
\begin{equation}
	f_*(\mathbf{z}^*):=\inf_{\mathbf{z}\in\R^d}\bigl\{\mathbf{z}^*\cdot\mathbf{z}-f(\mathbf{z})\bigr\},
\end{equation}
see \cite[Sections~12,~30]{Ro70B} and \cite{Ro74}. It satisfies $f_*=-(-f)^*$, so no separate theory is needed beyond sign changes. The \emph{concave superdifferential} is defined via $\partial f(\mathbf{z}):=-\partial(-f)(\mathbf{z})$, and satisfies the duality relation $\mathbf{z}^*\in\partial f(\mathbf{z})$ if and only if $f(\mathbf{z})+f_*(\mathbf{z}^*)=\mathbf{z}\cdot\mathbf{z}^*$.

\begin{definition}[Convex/concave Fenchel-Legendre transform]\label{def:genFc}
For $d\in \mathbb{N}$ and proper $f:\R^{d}\times \R^d\to [-\infty,+\infty]$, we define the convex/concave Fenchel-Legendre transform of $f$ with respect to the second variable as
\begin{align}\label{genFc}
	f^\diamond(\boldsymbol{\xi},\mathbf{z}^*):=\begin{cases}
		f^*(\boldsymbol{\xi},\mathbf{z}^*)&\text{if $\mathbf{z}\mapsto f(\boldsymbol{\xi},\mathbf{z})$ is convex},\\
		-(-f)^*(\boldsymbol{\xi},-\mathbf{z}^*)&\text{if $\mathbf{z}\mapsto f(\boldsymbol{\xi},\mathbf{z})$ is concave},
	\end{cases}
	\quad \text{for all }\boldsymbol{\xi},\mathbf{z}^*\in \R^d.
\end{align}
\end{definition}
\begin{remark}\label{rem:genFc}
When $f$ is differentiable, strictly convex or strictly concave and $\nabla f$ is continuous and invertible (invertibility actually automatically follows from strict convexity or strict concavity), the two definitions coincide: $f^\diamond$ given by \eqref{genFc} equals $f^\diamond$ given by \eqref{genFc-smooth}.

\end{remark}

\subsection{Necessary conditions for critical points or extremals}\label{secC:crit}

\subsubsection{Smooth case}

We begin with the smooth case, which provides the Euler-Lagrange equations that characterize constrained critical points in the sense of Definition~\ref{def:crit-smooth}.

\begin{proposition}[Euler-Lagrange equations: smooth case]\label{prop:EL-smooth}
Assume that
$\Phi_{\mathbf{y}}(\boldsymbol{\xi},\cdot)$ and $\widehat\Phi_{\mathbf{y}}(\boldsymbol{\xi},\cdot)$ are differentiable with respect to $\mathbf{b}$ and $\mathbf{m}$ respectively, for a.e.\ $\boldsymbol{\xi}\in\mathbf{y}(\Omega)$. In addition, assume that
the gradients $\nabla_{\mathbf{b}}\Phi_{\mathbf{y}}$ and $\nabla_{\mathbf{m}}\widehat\Phi_{\mathbf{y}}$ (and hence $\nabla_{\mathbf{m}}\widehat\Psi_{\mathbf{y}}$) satisfy a linear growth condition in their second argument, uniformly in $\boldsymbol{\xi}$, so that the functionals $\cE_{\mathbf{y}}$ and $\widehat\cE_{\mathbf{y}}$ are G\^ateaux-differentiable in their respective magnetic variables on $L^2$.
Then the following holds.
\begin{itemize}
\item[(i)] Let $(\overline{\mathbf{m}},\overline{\mathbf{h}}_{\rm s})\in L^2(\mathbf{y}(\Omega);\mathbb{R}^3)\times L^2(\mathbb{R}^3;\mathbb{R}^3)$ satisfy \eqref{divcon}. If
\begin{equation}
\frac{{\rm d}}{{\rm d}\varepsilon}\Big|_{\varepsilon=0}\widehat\cE_{\mathbf{y}}(\overline{\mathbf{m}}+\varepsilon\boldsymbol\eta,\, \overline{\mathbf{h}}_{\rm s}+\varepsilon\boldsymbol\zeta)=0
\end{equation}
for every perturbation direction $(\boldsymbol\eta,\boldsymbol\zeta)\in L^2(\mathbf{y}(\Omega);\mathbb{R}^3)\times L^2(\mathbb{R}^3;\mathbb{R}^3)$ such that 
$\operatorname{curl} \boldsymbol\zeta=0$ and $\operatorname{div} (\chi_{\mathbf{y}(\Omega)} \boldsymbol\eta+\boldsymbol\zeta)=0$,
then $(\overline{\mathbf{m}},\overline{\mathbf{h}}_{\rm s})$ is a smooth constrained critical point of $\widehat\cE_{\mathbf{y}}$ in the sense of Definition~\ref{def:crit-smooth}, i.e., \eqref{equil-mh-smooth} holds.

\item[(ii)] Let $\overline{\mathbf{b}}\in L^2(\mathbb{R}^3;\mathbb{R}^3)$ with $\operatorname{div}\overline{\mathbf{b}}=0$. If
\begin{equation}
\frac{{\rm d}}{{\rm d}\varepsilon}\Big|_{\varepsilon=0}\cE_{\mathbf{y}}(\overline{\mathbf{b}}+\varepsilon\boldsymbol{\psi})=0
\end{equation}
for every perturbation direction $\boldsymbol{\psi}\in L^2(\mathbb{R}^3;\mathbb{R}^3)$ such that $\operatorname{div}\boldsymbol{\psi}=0$, then $\overline{\mathbf{b}}$ is a smooth constrained critical point of $\cE_{\mathbf{y}}$ in the sense of Definition~\ref{def:crit-smooth}, i.e., there exists $\overline\varphi\in\mathring{W}^{1,2}(\mathbb{R}^3)$ such that \eqref{equil-b-smooth} holds.
\end{itemize}
\end{proposition}

\begin{proof}
\textit{Part (i).} 
The conditions $\operatorname{curl} \boldsymbol\zeta=0$ and $\operatorname{div} (\chi_{\mathbf{y}(\Omega)} \boldsymbol\eta+\boldsymbol\zeta)=0$ are equivalent to $\boldsymbol\zeta=-\cP[\chi_{\mathbf{y}(\Omega)}\boldsymbol\eta]$ by Lemma~\ref{lem:Helmholtz}, so the perturbation directions are exactly those of the form $(\boldsymbol\eta,-\cP[\chi_{\mathbf{y}(\Omega)}\boldsymbol\eta])$ for $\boldsymbol\eta\in L^2(\mathbf{y}(\Omega);\mathbb{R}^3)$. For any perturbation direction $\boldsymbol\eta\in L^2(\mathbf{y}(\Omega);\mathbb{R}^3)$, differentiate at $\varepsilon=0$ the map
$$\varepsilon\mapsto\widehat\cE_{\mathbf{y}}(\overline{\mathbf{m}}+\varepsilon\boldsymbol\eta,\, -\cP[\chi_{\mathbf{y}(\Omega)}(\overline{\mathbf{m}}+\varepsilon\boldsymbol\eta)]).$$
The derivative of the material term gives $\int_{\mathbf{y}(\Omega)}\nabla_{\mathbf{m}}\widehat\Phi_{\mathbf{y}}(\cdot,\overline{\mathbf{m}})\cdot\boldsymbol\eta\,{\rm d}\boldsymbol{\xi}$. The derivative of the stray-field term $\frac{\mu_0}{2}\int_{\mathbb{R}^3}|\cP[\chi_{\mathbf{y}(\Omega)}\mathbf{m}]|^2\,{\rm d}\boldsymbol{\xi}$ at $\overline{\mathbf{m}}$ in direction $\boldsymbol\eta$ is
$$\mu_0\int_{\mathbb{R}^3}\cP[\chi_{\mathbf{y}(\Omega)}\overline{\mathbf{m}}]\cdot\cP[\chi_{\mathbf{y}(\Omega)}\boldsymbol\eta]\,{\rm d}\boldsymbol{\xi} = -\mu_0\int_{\mathbf{y}(\Omega)}\overline{\mathbf{h}}_{\rm s}\cdot\boldsymbol\eta\,{\rm d}\boldsymbol{\xi},$$
where we used self-adjointness of $\cP$, the identity $\overline{\mathbf{h}}_{\rm s}=-\cP[\chi_{\mathbf{y}(\Omega)}\overline{\mathbf{m}}]$, and $\cP\overline{\mathbf{h}}_{\rm s}=\overline{\mathbf{h}}_{\rm s}$ (since $\overline{\mathbf{h}}_{\rm s}$ is already in the range of $\cP$). The derivative of the forcing term $-\int_{\mathbf{y}(\Omega)}\mathbf{m}\cdot\mathbf{b}_{\rm a}\,{\rm d}\boldsymbol{\xi}$ gives $-\int_{\mathbf{y}(\Omega)}\boldsymbol\eta\cdot\mathbf{b}_{\rm a}\,{\rm d}\boldsymbol{\xi}$. Setting the total derivative to zero for all $\boldsymbol\eta$ and using the definition $\widehat\Psi_{\mathbf{y}}=\widehat\Phi_{\mathbf{y}}+\frac{\mu_0}{2}|\mathbf{m}|^2$ gives \eqref{equil-mh-smooth}.

\textit{Part (ii).} By definition of $\cE_{\mathbf{y}}$ in \eqref{E(b)}, we have $\cE_{\mathbf{y}}(\mathbf{b})=\cF_{\mathbf{y}}(\mathbf{b})+\cG(\mathbf{b})$ for all $\mathbf{b}\in L^2(\mathbb{R}^3;\mathbb{R}^3)$, where $\cF_{\mathbf{y}}$ and $\cG$ are as in \eqref{E(b)}. The G\^ateaux derivative of $\cF_{\mathbf{y}}$ at $\overline{\mathbf{b}}$ in direction $\boldsymbol\psi$ is
$$D\cF_{\mathbf{y}}(\overline{\mathbf{b}})[\boldsymbol\psi]=\int_{\mathbb{R}^3}\left(\chi_{\mathbf{y}(\Omega)}\nabla_{\mathbf{b}}\Phi_{\mathbf{y}}(\cdot,\overline{\mathbf{b}})+\frac{\overline{\mathbf{b}}-\mathbf{b}_{\rm a}}{\mu_0}\right)\cdot\boldsymbol\psi\,{\rm d}\boldsymbol{\xi}.$$
Since $\overline{\mathbf{b}}$ is a local minimizer of $\cE_{\mathbf{y}}$ subject to $\operatorname{div}\mathbf{b}=0$, the first-order condition $D\cF_{\mathbf{y}}(\overline{\mathbf{b}})[\boldsymbol\psi]=0$ must hold for every divergence-free $\boldsymbol\psi\in L^2(\mathbb{R}^3;\mathbb{R}^3)$. The space of divergence-free fields is the kernel of $\cP$ (cf.\ Lemma~\ref{lem:Helmholtz}), so its orthogonal complement is the range of $\cP$, which consists of gradient fields $\nabla\varphi$. Hence there exists $\overline\varphi\in\mathring{W}^{1,2}(\mathbb{R}^3)$ such that
$$\chi_{\mathbf{y}(\Omega)}\nabla_{\mathbf{b}}\Phi_{\mathbf{y}}(\cdot,\overline{\mathbf{b}})+\frac{\overline{\mathbf{b}}-\mathbf{b}_{\rm a}}{\mu_0}+\nabla\overline\varphi=0\quad\text{a.e. in }\mathbb{R}^3,$$
which is exactly \eqref{equil-b-smooth}.
\end{proof}

\subsubsection{Convex/concave (nonsmooth) case}

We now turn to the nonsmooth setting. The following propositions show that the subdifferential inclusions of Definition~\ref{def:crit} are the natural generalization of the Euler-Lagrange equations above, characterizing global minimizers in the case of convex functionals. \bbb

In this section, we briefly demonstrate that like any classical notion of critical points, our Definition~\ref{def:crit} of magnetic critical points extrapolates the natural first order necessary condition for local extrema, here in case of particular constrained functionals. The necessity of conditions like \eqref{equil-b} and \eqref{equil-mh} for local minimizers of nonsmooth functionals is apparently not as comprehensively studied in the literature as one might like, except for single integrals and for integrands which are a convex function of gradients in the leading order term (see, e.g., \cite{BonCel10a}).
To the best of our knowledge, the only available results beyond that, in particular for constrained problems, were given in \cite{Dol14a,Dol21a} but are not immediately applicable either. 
Fortunately, our situation is relatively simple and does not require overly technical proofs.
We, too, will restrict our attention to convex integrands here, and fully concave cases are of course analogous. 
The results below could be suitably generalized to semiconvex of semiconcave integrands, \bbb but the discussion of anything beyond that would require significantly more effort.

\begin{proposition}[Characterization of magnetic critical points of convex $\widehat{\mathcal E}$]\label{prop:charact-mh}
Suppose that \eqref{H1} and \eqref{H2} hold with $\widehat\Phi_{\mathbf{y}}$ convex in $\mathbf{m}$, that
 $(\overline{\mathbf{m}},\overline{\mathbf{h}}_{\rm s})\in L^2(\Omega;\R^3)\times L^2(\R^3;\R^3)$ satisfies the constraints \eqref{divcon} and that $\widehat\cE(\mathbf{y},\overline{\mathbf{m}},\overline{\mathbf{h}}_{\rm s})<\infty$.
If  $(\overline{\mathbf{m}},\overline{\mathbf{h}}_{\rm s})$ is a constrained local minimizer of $(\mathbf{m},\mathbf{h}_{\rm s})\mapsto \widehat{\mathcal E}(\mathbf{y},\mathbf{m},\mathbf{h}_{\rm s})$ subject to \eqref{divcon},
then  $(\overline{\mathbf{m}},\overline{\mathbf{h}}_{\rm s})$ is a magnetic critical point of $\widehat\cE$ in the sense of Definition~\ref{def:crit}, i.e., \eqref{equil-mh} holds.
Conversely, if 
$(\overline{\mathbf{m}},\overline{\mathbf{h}}_{\rm s})$ satifies \eqref{equil-mh}, 
	then $(\overline{\mathbf{m}},\overline{\mathbf{h}}_{\rm s})$ is a constrained global minimizer of $(\mathbf{m},\mathbf{h}_{\rm s})\mapsto \widehat{\mathcal E}(\mathbf{y},\mathbf{m},\mathbf{h}_{\rm s})$.
\end{proposition}
\begin{proof}
First notice that by our assumptions, $(\mathbf{m},\mathbf{h}_{\rm s})\mapsto\widehat{\mathcal E}(\mathbf{y},\mathbf{m},\mathbf{h}_{\rm s})$ is convex.
Hence, any constrained local minimizer $(\overline{\mathbf{m}},\overline{\mathbf{h}}_{\rm s})$ of $(\mathbf{m},\mathbf{h}_{\rm s})\mapsto\widehat{\mathcal E}(\mathbf{y},\mathbf{m},\mathbf{h}_{\rm s})$ with finite energy
is automatically a constrained global minimizer and it suffices to discuss global minimizers.

Given \eqref{divcon}, $\mathbf{h}_{\rm s}=-\cP[\chi_{\mathbf{y}(\Omega)}\mathbf{m}]$ by Remark~\ref{rem:h-from-m}. By the definition of $\widehat\cE$ in \eqref{E(mh)}, we therefore have that
\begin{align}
		\begin{aligned}
			 \widehat\cE(\mathbf{y},\mathbf{m},\mathbf{h}_{\rm s})=&\widehat\cF_{\mathbf{y}}(\mathbf{m})+\widehat\cG(\mathbf{m}),
		\end{aligned}
\end{align}
where
\begin{align}
		\begin{aligned}
			 \widehat\cF_{\mathbf{y}}(\mathbf{m}):=\int_{\mathbf{y}(\Omega)} \big(\widehat\Phi_{\mathbf{y}}(\boldsymbol{\xi},\mathbf{m})- \mathbf{m} \cdot \mathbf{b}_{\rm a}\big)\,{\rm d}\boldsymbol{\xi}~~\text{and}~~\widehat\cG(\mathbf{m}):=\int_{\mathbb R^3}\frac{\mu_0}{2}\big|\cP[\chi_{\mathbf{y}(\Omega)}\mathbf{m}]\big|^2 \,{\rm d}\boldsymbol{\xi}.
		\end{aligned}
\end{align}
Since $\cP$ is a bounded self-adjoint linear projection operator on $L^2(\R^3;\R^3)$, $\widehat\cG_{\mathbf{y}}:L^2(\R^3;\R^3)\to \R$ is convex and Frechet-differentiable with derivative
\begin{equation}
	D\widehat{\cG}_{\mathbf{y}}(\mathbf{m})[\dot{\mathbf{m}}]=\int_{\mathbb R^3}\mu_0\cP[\chi_{\mathbf{y}(\Omega)}\mathbf{m}]\cdot \cP[\chi_{\mathbf{y}(\Omega)}\dot{\mathbf{m}}] \,{\rm d}\boldsymbol{\xi}=\int_{\mathbf{y}(\Omega)}\mu_0\cP[\chi_{\mathbf{y}(\Omega)}\mathbf{m}]\cdot \dot{\mathbf{m}}\,{\rm d}\boldsymbol{\xi},
\end{equation}
for all $\dot{\mathbf{m}}\in L^2(\mathbf{y}(\Omega);\R^3)$.
Moreover, $\widehat\cF_{\mathbf{y}}$ is convex and its subdifferential is given by
\begin{align}
	\partial\widehat\cF_{\mathbf{y}}(\mathbf{m})=\left\{ \dot{\mathbf{m}}\mapsto \int_{\mathbf{y}(\Omega)} \mathbf{g}\cdot \dot{\mathbf{m}}\,{\rm d}\boldsymbol{\xi} ~\left|~
	\begin{aligned}
		&\mathbf{g}\in L^2(\mathbf{y}(\Omega);\R^3),\\
		&\mathbf{g}\in \mu_0\cP[\chi_{\mathbf{y}(\Omega)}\mathbf{m}]-\mathbf{b}_{\rm a}+\partial_{\mathbf{m}} \widehat\Phi_{\mathbf{y}}(\cdot,\mathbf{m}(\boldsymbol{\xi}))	~~\text{a.e.~on }\mathbf{y}(\Omega)
	\end{aligned}	
	\right.\right\}
\end{align}
Here, we have used Riesz representation to identify the elements of the subdifferential naturally given by continuous linear functionals in the dual of $L^2(\mathbf{y}(\Omega);\R^3)$ 
with functions $\mathbf{g}\in L^2(\mathbf{y}(\Omega);\R^3)$.
Combined, we see that $\mathbf{m}\mapsto \widehat\cF_{\mathbf{y}}(\mathbf{m})+\widehat\cG(\mathbf{m})$, $L^2(\mathbf{y}(\Omega);\R^3)\to \R$, is convex with subdifferential
\begin{align}
	\partial(\widehat\cF_{\mathbf{y}}+\widehat\cG)(\mathbf{m})=\left\{ \dot{\mathbf{m}}\mapsto \int_{\mathbf{y}(\Omega)} \mathbf{g}\cdot \dot{\mathbf{m}}\,{\rm d}\boldsymbol{\xi} ~\left|~
	\begin{aligned}
		&\mathbf{g}\in L^2(\mathbf{y}(\Omega);\R^3),\\
		&\mathbf{g}\in \mu_0\cP[\chi_{\mathbf{y}(\Omega)}\mathbf{m}]-\mathbf{b}_{\rm a}+\partial_{\mathbf{m}} \widehat\Phi_{\mathbf{y}}(\cdot,\mathbf{m})~~\text{a.e.~on }\mathbf{y}(\Omega) 
	\end{aligned}	
	\right.\right\}.
\end{align}
Finally, $(\overline{\mathbf{m}},\overline{\mathbf{h}}_{\rm s})$ is a constrained global minimizer $(\overline{\mathbf{m}},\overline{\mathbf{h}}_{\rm s})$ of $(\mathbf{m},\mathbf{h}_{\rm s})\mapsto \widehat{\mathcal E}(\mathbf{y},\mathbf{m},\mathbf{h}_{\rm s})$
if and only if $\overline{\mathbf{m}}$ is a global minimizer of $\widehat\cF_{\mathbf{y}}+\widehat\cG$. 
The latter is in turn equivalent to 
\begin{equation}
	\mathbf{0}\in \partial (\widehat\cF_{\mathbf{y}}+\widehat\cG).
\end{equation}
By our characterizations of $\partial \cF_{\mathbf{y}}$ and $\partial \cG$, we see that this in turn is equivalent 
to \eqref{equil-mh}.
\end{proof}

\begin{proposition}[Characterization of magnetic critical points of convex $\mathcal E$]\label{prop:charact-b}
Suppose that \eqref{H1p} and \eqref{H2p} hold, that
\begin{align}
	\Psi_y(\boldsymbol{\xi},\mathbf{b}) =\chi_{y(\Omega)}(\boldsymbol{\xi})\Phi_{\mathbf{y}}(\boldsymbol{\xi},\mathbf{b})+\frac{1}{2\mu_0}|\mathbf{b}-\mathbf{b}_{\rm a}(\boldsymbol{\xi})|^2
	\qquad\text{is convex in $\mathbf{b}$ for a.e.~$\boldsymbol{\xi}$},
\end{align} \bbb
that
$\overline{\mathbf{b}}\in L^2(\R^3;\R^3)$ satisfies $\operatorname{div} \overline{\mathbf{b}}=0$ on $\R^3$
and that $\cE(\mathbf{y},\overline{\mathbf{b}})<\infty$.
If $\overline{\mathbf{b}}$ is a local minimizer of $\mathbf{b} \mapsto \mathcal{E}(\mathbf{y}, \mathbf{b})$, $L^2(\R^3;\R^3)\to \R$ under the constraint $\operatorname{div} \mathbf{b}=0$, then 
$\overline{\mathbf{b}}$ is a magnetic critical point of $\cE$ in the sense of Definition~\ref{def:crit}, i.e., \eqref{equil-b} holds 
with some $\overline{\varphi}\in \mathring W^{1,2}(\mathbb R^3)$.
Conversely, if 
\eqref{equil-b} holds for some $\overline{\varphi}\in \mathring W^{1,2}(\mathbb R^3)$, then $\overline{\mathbf{b}}$ is a constrained global minimizer of $\mathbf{b}\mapsto \mathcal E(\mathbf{y},\mathbf{b})$.
\end{proposition} 
\begin{remark} Unless $\Phi_y$ is either convex or concave in $\mathbf{b}$, the subdifferential (or superdifferential) appearing in \eqref{equil-b} might be empty on a set of positive measure, rendering \eqref{equil-b} useless. 
However, in both good cases, we have that 
\begin{align}\label{EL-b-PhiPsi}
	\partial_{\mathbf{b}}\Psi_{\mathbf{y}}(\boldsymbol{\xi},\mathbf{b})=\frac{\mathbf{b}-\mathbf{b}_{\rm a}(\boldsymbol{\xi})}{\mu_0}+\partial_{\mathbf{b}}\Phi_{\mathbf{y}}(\boldsymbol{\xi},\mathbf{b})
	\quad\text{for every $\mathbf{b}\in \R^3$ and a.e.~$\boldsymbol \xi\in \mathbf{y}(\Omega)$ }
\end{align}
with the usual convex subdifferential applied to $\Psi_{\mathbf{y}}$. Using this to substitute $\partial_{\mathbf{b}}\Phi_{\mathbf{y}}$, \eqref{equil-b} remains meaningful also under our present assumptions.
\end{remark} \bbb
\begin{proof}[Proof of Proposition~\ref{prop:charact-b}]
Since $\mathbf{b} \mapsto \mathcal{E}(\mathbf{y}, \mathbf{b})$ is convex, it suffices to discuss global minimizers, because any local minimizers with finite energy are automatically global.
By definition of $\cE$ in \eqref{E(b)}, we have that
\begin{align}
		\begin{aligned}
			 \cE(\mathbf{y},\mathbf{b})=&\cF_{\mathbf{y}}(\mathbf{b})+\cG(\mathbf{b})\quad\text{for all $\mathbf{b}\in L^2(\R^3;\R^3)$ with $\div \mathbf{b}=0$ on $\R^3$},
		\end{aligned}
\end{align}
where $\cF_{\mathbf{y}},\cG:L^2(\R^3;\R^3)\to \R\cup\{+\infty\}$,
\begin{align}
			 \cF_{\mathbf{y}}(\mathbf{b}):=\int_{\R^3} \Psi_y(\boldsymbol{\xi},\mathbf{b})\,{\rm d}\boldsymbol{\xi},~~~
			\cG(\mathbf{b}):=\begin{cases}
				0 & \text{if $\div \mathbf{b}=0$ on $\R^3$,}\\
				+\infty & \text{else.}
		\end{cases}
\end{align}
Both $\cF_{\mathbf{y}}$ and $\cG$ are convex and lower semicontinuous, the latter because the set of divergence-free fields in $L^2$ is closed.
Moreover, for any $\mathbf{b}\in L^2(\R^3;\R^3)$,
their subdifferentials (as unconstrained functionals on $L^2(\R^3;\R^3)$) are given by
\begin{align}
\begin{alignedat}{2}
	&\partial\cF_{\mathbf{y}}(\mathbf{b})&&=\left\{ \dot{\mathbf{b}}\mapsto \int_{\R^3} \mathbf{g}\cdot \dot{\mathbf{b}}\,{\rm d}\boldsymbol{\xi} ~\left|~
	\begin{aligned}
		&\mathbf{g}\in L^2(\R^3;\R^3),\\
		&\mathbf{g}(\boldsymbol{\xi})\in \partial_{\mathbf{b}} \Psi_{\mathbf{y}}(\boldsymbol{\xi},\mathbf{b}(\boldsymbol{\xi}))
		~~\text{for a.e. }\boldsymbol{\xi}\in \R^3
	\end{aligned}	
	\right.\right\},\\
		&\partial\cG(\mathbf{b})&&=\left\{\left. \dot{\mathbf{b}}\mapsto \int_{\mathbf{y}(\Omega)} \mathbf{g}\cdot \dot{\mathbf{b}}\,{\rm d}\boldsymbol{\xi} ~\right|~
		\mathbf{g}=\nabla\overline{\varphi}~\text{for a }\overline{\varphi} \in \mathring W^{1,2}(\R^3)
	\right\}\quad \text{if }\div \mathbf{b}=0,\\ \bbb
	  & \partial\cG(\mathbf{b})&&=\emptyset\quad \text{otherwise}. \bbb
\end{alignedat}
\end{align}
Here, we have used Riesz representation to identify the elements of the subdifferential naturally given by linear functionals in the dual of $L^2(\R^3;\R^3)$ 
with functions $\mathbf{g}\in L^2(\R^3;\R^3)$. Moreover, for stating $\partial\cG(\mathbf{b})$, 
we have exploited that 
\begin{equation}
	\mathbf{g}=\nabla\overline{\varphi}~\text{for a }\overline{\varphi}\in \mathring W^{1,2}~~~\text{if and only if}~~~\int_{\R^3} \mathbf{g}\cdot \dot{\mathbf{b}}\,{\rm d}\boldsymbol{\xi}=0~~\text{for all $\dot{\mathbf{b}}\in L^2$ with $\div \dot{\mathbf{b}}=0$}.
\end{equation}
(In terms of the projection of Appendix~\ref{secA:Helmholtz}: the range of $\cP$, i.e., the gradient fields, is the orthogonal complement of its kernel, i.e., the divergence-free fields.)

Finally, we observe that $\overline{\mathbf{b}}$ is a constrained minimizer of $\mathbf{b} \mapsto \mathcal{E}(\mathbf{y}, \mathbf{b})$ if and only if 
$\overline{\mathbf{b}}$ is an (unconstrained) minimizer of $\cF_{\mathbf{y}}+\cG$. In view of the definition $\cG$, the latter is equivalent to
\begin{equation}\label{neccondE}
	\div \overline{\mathbf{b}}=0\quad\text{and}\bbb\quad \mathbf{0}\in \partial (\cF_{\mathbf{y}}+\cG)(\overline{\mathbf{b}})=\partial \cF_{\mathbf{y}}(\overline{\mathbf{b}})+\partial \cG(\overline{\mathbf{b}}).
\end{equation}
By our characterizations of $\partial \cG$ and $\partial \cF_{\mathbf{y}}$, we see that the inclusion in \eqref{neccondE} is in turn equivalent 
to the existence of some $\overline{\varphi}\in \mathring W^{1,2}(\mathbb R^3)$ such that
\eqref{equil-b} holds (recall \eqref{EL-b-PhiPsi}).
\end{proof}

\bibliographystyle{abbrv}
\bibliography{references}

\end{document}